\documentclass[reqno,onecolumn,letterpaper,10pt]{amsart}
\usepackage{amsmath,amssymb,fullpage,calc,enumerate}
\usepackage{graphicx,subfigure,psfrag,multicol,tikz}
\usepackage{algorithm,algorithmicx,array}
\usepackage{algpseudocode}
\usepackage{booktabs, hyperref}
\usepackage{color}
\usepackage{xcolor}
\usepackage{rotating}
\usepackage[title]{appendix}

\hypersetup{ 
    colorlinks,
    linkcolor={red!50!black},
    citecolor={blue!70!black},
    urlcolor={blue!80!black}
}
\usetikzlibrary{arrows.meta}
\usetikzlibrary{decorations.shapes}
\usetikzlibrary{decorations.markings}
\usetikzlibrary{calc,intersections,through,backgrounds,scopes}
\usetikzlibrary{decorations.pathmorphing}

\allowdisplaybreaks[1] \numberwithin{equation}{section}
\newtheorem{theorem}{Theorem}[section]

\newtheorem{lemma}[theorem]{Lemma}
\newtheorem{corollary}[theorem]{Corollary}

\theoremstyle{definition}

\def\RR{\mathbb{R}}
\def\ZZ{\mathbb{Z}}

\newcommand{\vect}[1]{\boldsymbol{\mathbf{#1}}}
\usetikzlibrary{fit}
\usetikzlibrary{shapes}
\usetikzlibrary{shapes.geometric,positioning}
\tikzstyle{vrtx} =[circle,fill=black,draw=black,inner sep=0pt,minimum size=3pt]
\tikzstyle{labvrtx}=[circle,fill=white,draw=black,inner sep=0.02cm, minimum size=0.5cm]
{\begin{list}{}{%
   \settowidth{\labelwidth}{\textbf{#1:}}%
   \setlength{\parsep}{0.0cm}%
   \setlength{\itemsep}{0.10cm}%
   \setlength{\leftmargin}{\labelwidth+\labelsep}}}%
{\end{list}}

{\begin{list}{}{%
   \settowidth{\labelwidth}{\textsf{#1)}}%
   \setlength{\parsep}{0.0cm}%
   \setlength{\itemsep}{0.10cm}%
   \setlength{\leftmargin}{\labelwidth+\labelsep}}}%
{\end{list}}
\newcommand{\myred}[1]{\textcolor{black}{#1}}
\newcommand{\pu}[1]{\textcolor{black}{#1}}
\usepackage[normalem]{ulem}

\tikzset{My Style/.style={shape=circle,draw,inner sep=0.0cm, minimum size=3pt, fill=black}}

\title{A class of exponential neighbourhoods for the quadratic travelling salesman problem}

\author[Brad D. Woods, Abraham P. Punnen]{Brad D. Woods, Abraham P. Punnen}

\date{\today}
\thanks{This work was supported by NSERC discovery grants awarded to Abraham P. Punnen}
\thanks {Brad Woods, Department of Mathematics, Simon Fraser University, 8888 University Drive, Burnaby, British Columbia, V5A 1S6, Canada. {\tt bdw2@sfu.ca}}
\thanks {Abraham Punnen, Department of Mathematics, Simon Fraser University, 8888 University Drive, Burnaby, British Columbia, V5A 1S6, Canada. {\tt apunnen@sfu.ca}}

\begin{document}

\begin{abstract}
The Quadratic Travelling Salesman Problem (QTSP) is to find a least-cost Hamiltonian cycle in an edge-weighted graph, where costs are defined on all pairs of edges such that each edge in the pair is contained in the Hamiltonian cycle.  This is a more general version than the one that appears in the literature as the QTSP, denoted here as the \emph{adjacent quadratic} TSP, which only considers costs for pairs of adjacent edges.  Major directions of research work on the linear TSP include exact algorithms, heuristics, approximation algorithms, polynomially solvable special cases and exponential neighbourhoods~\cite{Gutin2002-2} among others.  In this paper we explore the complexity of searching exponential neighbourhoods for QTSP, the fixed-rank QTSP, and the adjacent quadratic TSP.  The fixed-rank QTSP is introduced as a restricted version of the QTSP where the cost matrix has fixed rank $p$.  It is shown that fixed-rank QTSP is solvable in pseudopolynomial time and admits an FPTAS for each of the special cases studied, except for the case of matching edge ejection tours. The adjacent quadratic TSP is shown to be polynomially-solvable in many of the cases for which the linear TSP is polynomially-solvable. Interestingly, optimizing over the matching edge ejection tour neighbourhood is shown to be pseudopolynomial for the rank 1 case without a linear term in the objective function, but NP-hard for the adjacent quadratic TSP case. \pu{We also show that the quadratic shortest path problem on an acyclic digraph can be solved in pseudopolynomial time and by an FPTAS when the rank of the associated cost matrix is fixed.}
\end{abstract}

\date{\today}
\maketitle


\section{Introduction}
The Travelling Salesman Problem (TSP) is to find a least-cost Hamiltonian cycle in an edge-weighted graph.  It is one of the most widely studied hard combinatorial optimization problems.  The TSP has been used to model a wide variety of applications.  For details we refer the reader to the well-known books~\cite{Applegate2011,Cook2012, Gutin2002-2,Lawler1985,  Reinelt1994}.  For clarity of discussion, we will refer to this problem as the \emph{linear TSP}.

Let $G=(V,E)$ be an undirected graph on the vertex set $V=\{1,\ldots ,n\}$ and edge set $E=\{1,2,\ldots,m\}$. For each edge $e \in E$, a cost $c(e)$ is given.  Also, for each pair of edges $(e,f)$, another cost $q(e,f)$ is prescribed.  Let $\mathcal{F}$ be the set of all Hamiltonian cycles (tours) in $G$.  The cost $f(\tau)$ of a tour $\tau \in \mathcal{F}$ is given by
\[
    f(\tau) = \sum \limits_{(e,f)\in \tau \times \tau} q(e,f) + \sum \limits_{e\in \tau} c(e).
\]
Then the quadratic travelling salesman problem (QTSP) is to find a least cost tour $\tau\in\mathcal{F}$ such that $f(\tau)$ is as small as possible.

The problem QTSP has received only very limited attention in literature.  A special case of QTSP has been studied by various authors recently \cite{fischer2014analysis, fischer2016polyhedral, fischer2014exact, fischer2013symmetric, Jager:2008, rostami2016lower} where $q(i,j)$ is assumed to be zero if edges $i$ and $j$ are not adjacent.  Although this restricted problem is known as the quadratic TSP in literature, to distinguish it from the general problem, we refer to it as the \emph{adjacent quadratic TSP}, which we denote by QTSP(A).  The $k$-neighbour TSP studied by Woods et al.~\cite{woods2010generalized} is also related to QTSP.  The linear TSP on Halin graphs was studied in~\cite{Cornuejols83}, and an $O(n)$ algorithm was given.  In~\cite{Woods2017Halin} it is shown that QTSP on Halin graphs is strongly NP-hard, however, an $O(n)$ algorithm solves QTSP(A) on this class of graphs.  The linear TSP is solvable in polynomial time when the set of tours is restricted to pyramidal tours~\cite{burkard1999}.  In~\cite{Woods2017PQ}, it is shown that QTSP over the set of pyramidal tours in strongly NP-hard but the corresponding QTSP(A) can be solved in $O(n^3)$ over this class of tours.

Let $Q$ be the $m$ by $m$ matrix with $(e,f)$th element $q(e,f)$, for $e,f\in E$.  If the rank of $Q$ is $p$, then by using the rank decomposition of $Q$, QTSP can be written in another form as
\begin{align*}
\text{Minimize     }  q(\tau) =& \sum_{r=1}^p\left[\left(\sum_{e\in \tau} a_e^r\right)\left(\sum_{e\in \tau}b_e^r\right)\right] + \sum_{e\in \tau}c(e) \\
\text{Subject to } & \tau \in \mathcal{F}.
\end{align*}

For the general QTSP, we can eliminate the linear term by adding $c(e)$ to $q(e,e)$.  However, for the rank-restricted case, we need to consider the linear term explicitly since adding $c(e)$ to $q(e,e)$ could change the rank.  The variation where the linear term is absent is called homogeneous rank $p$ QTSP which is denoted by QTSP(p,H) and the general rank $p$ QTSP is denoted by QTSP(p,c).  It is easy to verify that QTSP(p,c) belongs to the class QTSP(p+1,H).  QTSP(p,c) and QTSP(p,H) restricted to pyramidal tours are studied in~\cite{Woods2017PQ}, to Halin graphs in~\cite{Woods2017Halin}, and is shown to be solvable in pseudopolynomial time when $p$ is fixed, and additionally, admits an FPTAS when the costs are non-negative.\\

Since TSP is a special case of QTSP(p,H), QTSP(p,c), QTSP(A) and QTSP, all these problems are strongly NP-hard~\cite{gj}.\\

Combinatorial optimization problems with the objective function as the product of two linear functions have been studied by Goyal et al.~\cite{Goyal2011} and Kern and Woeginger~\cite{kern2007quadratic}.  Mittal and Schulz~\cite{Mittal2013} considered a further general class of problems that subsumes the class of combinatorial optimization problems with objective functions as fixed sum of product\myred{s} of linear terms.  Thus, QTSP(1,H) falls under the general class considered in ~\cite{Goyal2011,kern2007quadratic} and QTSP(p,c) falls under the class considered in~\cite{Mittal2013}.  However, the corresponding results are not applicable to  QTSP(p,c) because the conditions imposed in deriving their results are not applicable to QTSP(p,c), even if $p=1$.

An instance of QTSP with cost matrix $Q$ is said to be linearizable if there exists an instance of the linear TSP with cost matrix $C$ such that for each tour, the QTSP and linear TSP objective function values are identical.  The corresponding QTSP linearization problem is studied in~\cite{Woods2017lin} and necessary and sufficient conditions are obtained for a cost matrix $Q$ to be linearizable.

Major directions of research work on the linear TSP include exact algorithms, heuristics, approximation algorithms, polynomially solvable special cases and exponential neighbourhoods~\cite{Gutin2002-2} among others.  In this paper we explore the complexity of searching \pu{some specific} exponential neighbourhoods for QTSP, QTSP(p,H) and QTSP(A).  Our focus is on exponential neighbourhoods that are studied in literature for the linear TSP and are known to be polynomially searchable. \pu{The study offers additional insights into the structure of QTSP, QTSP(p,H), and QTSP(A) and enhance our understanding of the boundary between polynomially solvable cases and NP-hard instances of these models.}
In particular, we consider the following classes of exponential neighborhoods (which will be described in detail later):
\begin{enumerate}
    \item Single edge ejection tours (SEE-tours) on a graph $G^*$ \cite{glover1997travelling},
    \item Double edge ejection tours (DEE-tours) on a graph $G^*$ \cite{glover1997travelling},
    \item Paired vertex graphs (PV-tours), and
    \item Matching edge ejection tours (MEE-tours) \cite{bentley1992fast, d1, d2, d3, Punnen2001}.
\end{enumerate}
Unlike the linear TSP, QTSP is strongly NP-hard for all these classes of tours.  Interestingly, the special cases of QTSP(A) are polynomially solvable for three out of four of these classes while QTSP(p,H) admits fully polynomial time approximation schemes (FPTAS).  Our complexity results are summarized in the following table.
\renewcommand{\arraystretch}{1.5}
\begin{center}
\begin{table}[ht] \caption{Summary of complexity results.}
\begin{tabular}{|r| l |cccc |}
	\hline
	Section 	& Neighbourhood			& QTSP 	  		    & QTSP(1,H)	& QTSP(p,H) 	& QTSP(A)       \\ \hline
	2		& SEE-tours on $G^*$     	& strongly NP-hard     & NP-hard,		& NP-hard,     	& $O(n^2)$      \\
	& & & FPTAS & FPTAS & \\
	3		& DEE-tours on $G^*$     	& strongly NP-hard     & NP-hard,  		& NP-hard,       	& $O(n^3)$ 	    \\
	& & & FPTAS & FPTAS & \\
	4		& Paired vertex graphs   	& strongly NP-hard     & NP-hard,  		& NP-hard,       	& $O(n)$ 		\\
	& & & FPTAS & FPTAS & \\	
	5		& MEE-tours  			& strongly NP-hard 	& NP-hard, 		&NP-hard,        	    & strongly NP-hard \\
	& & & FPTAS & FPTAS? & \\	
	 \hline
\end{tabular} \label{tab:summary}
\end{table}
\end{center}
\renewcommand{\arraystretch}{1}

In addition to their theoretical interest, exponential neighbourhoods are vital to the development of efficient very large-scale neighbourhood search (VLSN search) algorithms~\cite{Ahuja2002} and variable neighbourhood search algorithms~\cite{Mladenovic1997} \pu{ for a variety of hard combinatorial optimization problems}. The success of a variable neighborhood search algorithm depends on the availability of simple neighborhoods that can be searched using very fast algorithms and more powerful neighborhoods, often of exponential size, that can be searched using `reasonable algorithms', if not polynomial, to get the search \pu{process} out of entrapments at local optima of simpler neighborhoods.  In this sense, our study also contributes to the design of effective metaheuristics for QTSP(p,c) and QTSP(A).  Some areas of applications of the QTSP model include modelling the permuted variable length Markov model in bioinformatics~\cite{Jager:2008} as well as an optimal routing problem for unmanned aerial vehicles (UAVs)~\cite{woods2014}. Orlin, Punnen and Schulz~\cite{orlin2004approximate} showed that for \myred{any} linear combinatorial optimization problem \myred{with every solution having non-negative cost}, if a $\delta$-optimum over a neighborhood can be computed efficiently, then a $(\delta  + \epsilon)$ -local optimum for that neighborhood can be obtained efficiently. In~\cite{p2}, this result is extended to \myred{a} more general class of combinatorial optimization problems with some non-linear type objective functions. Exploiting these results, the FPTAS we obtained for QTSP(p,H) over various exponential neighborhoods can be utilized to construct fully-polynomial local optimization schemes~\cite{orlin2004approximate} for QTSP(p,H) with respect to the corresponding neighborhoods.

Throughout this paper, we use the following conventions.  All matrices will be denoted by capital letters, and all vectors with bold characters.  The $i$th component of a vector $\vect{a}$ is $a_i$.  Paths are represented by tuples of vertices, i.e. $(v_1,v_2,v_3)$ is a path through vertices $v_1,v_2$ then $v_3$.  Similarly, cycles are represented as ordered sets with the first and last elements the same, i.e. $(v_1,v_2,v_3,v_1)$ is the cycle which passes through $v_1$, $v_2$ and $v_3$ (in that order). \pu{For grapth theoretic terminology and notations we refer to~\cite{bondy1976}}

\section{Single edge ejection tours on $G^*$} \label{sec:SEE}

In this section we consider a special class of tours, called \emph{single edge ejection tours} (SEE-tours), introduced by Glover and Punnen~\cite{glover1997travelling}. We now present various complexity results regarding QTSP and its variations, restricted to this class.

An SEE-tour is defined using a graph $G^*=(V,E)$ which is a spanning subgraph of $K_n$.  Partition the vertex set of $K_n$ into a single vertex $t$, called the {\it tip vertex} and sets $V^1,V^2,\ldots,V^m$, \myred{$m\geq 2$,} such that $V^k =\{v^k_1,v^k_2,\ldots ,v^k_{r_k}\}$ and $|V^k|=r_k \geq 3$, for all $k=1,2,\ldots,m$.  Create a cycle $C(k)=(v^k_1,v^k_2,\ldots ,v^k_{r_k},v^k_1)$ for each $k=1,2,\ldots ,m$ and connect each vertex in $V^k$ to each vertex in $V^{k+1}$ by edges, for $k=1,2,\ldots,m-1$.  Let $E^k$ be the collection of edges so obtained for $k = 1,2,\ldots ,m-1$.  Add all possible edges from $t$ to each vertex in $V^1$ and $V^m$. Let $E^0$ be the set of edges joining $t$ and $V^1$,  and $E^{m}$ be the set of edges joining $t$ to $V^m$. The resulting graph is denoted by $G^*=(V,E)$. (See Figure \ref{fig:G-Star} for an example of a $G^*$ graph).

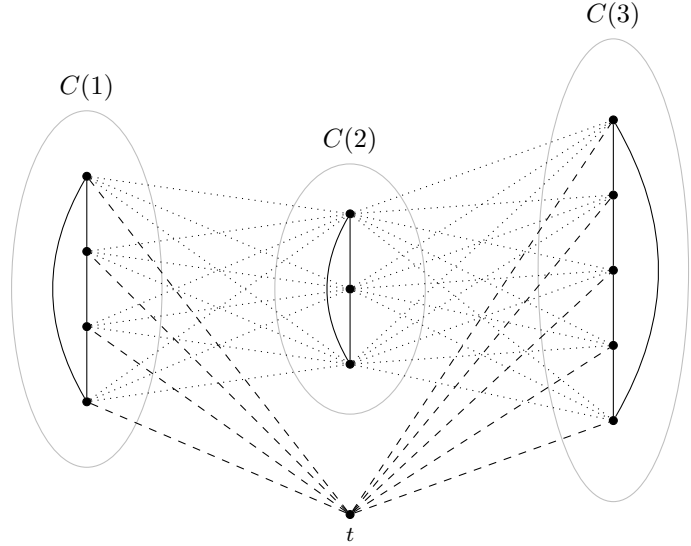
\begin{figure}[H]
\centering
\begin{tikzpicture}[scale=0.6]
\foreach \x[count=\xi] in {0.5,1.5,...,4}{
    \node[vrtx]  (v\xi) at (2,\x) {};}
\node[fit=(v4) (v3) (v2) (v1),ellipse,draw=gray!50,minimum width=2cm, label=above:{$C(1)$}] {};
\foreach \x[count=\xi] in {0.5,1.5,...,3}{
    \node[vrtx]  (w\xi) at (5.5,.5+\x) {};}
\node[fit=(w3) (w2) (w1),ellipse,draw=gray!50,minimum width=2cm, label=above:{$C(2)$}] {};
\foreach \x[count=\xi] in {0.5,1.5,...,5}{
    \node[vrtx]  (u\xi) at (9,-.25+\x) {};}
\node[fit=(u1) (u2) (u3) (u4) (u5),ellipse,draw=gray!50,minimum width=2cm, label=above:{$C(3)$}] {};
\draw (v1) -- (v2) -- (v3) -- (v4);
\draw[bend right] (v4) to (v1);
\draw (w1) -- (w2) -- (w3);
\draw[bend right] (w3) to (w1);
\draw (u1) -- (u2) -- (u3) -- (u4) -- (u5);
\draw[bend left] (u5) to (u1);
\foreach \x in {1,2,3,4}{
    \foreach \y in {1,2,3}{
        \draw[dotted] (v\x) -- (w\y);}}
\foreach \x in {1,2,3}{
    \foreach \y in {1,2,3,4,5}{
        \draw[dotted] (w\x) -- (u\y);}}
\node[vrtx, label=below:{\footnotesize $t$}]  (t) at (5.5,-1.0) {};
\foreach \y in {1,2,3,4,5}{
    \draw[dashed] (t) -- (u\y);}
\foreach \y in {1,2,3,4}{
    \draw[dashed] (t) -- (v\y);}
\end{tikzpicture}
	\caption{A graph $G^*$ with $n=13$ and $m=3$.}
	   \label{fig:G-Star}
\end{figure}

The travelling salesman problem on $G^*$ is known to be NP-hard~\cite{glover1997travelling} and it follows immediately that QTSP, QTSP(p,c), and QTSP(A) are all NP-hard on $G^*$.  Let us now consider a family of tours in $G^*$, called {it single edge ejection tours} (SEE-tours), which consists of all tours in $G^*$ which can be obtained by the following steps~\cite{glover1997travelling}.
\begin{enumerate}
    \item Choose an edge $(t,v^1_j)$ from $t$ to the cycle $C(1)$ and eject an edge $(v^1_j,v^1_{i})$ from $C(1)$.  The result creates a chain $P(1)$ from $t$ to $v^1_{i}$ which includes all edges of $C(1)$ except for the ejected edge.
    \item For each $k$ from $2$ to $m$, introduce the edge $(v^{k-1}_{i},v^k_j)$ from the vertex $v^{k-1}_{i}$ which is the end vertex of the chain $P(k-1)$ to the cycle $C(k)$, and eject an edge $(v^k_j,v^k_{i})$ from $C(k)$, where $i=j+1$ or $j-1$ modulo $r_k$, to create chain $P(k)$ from $t$ to $v^k_{i}$.
    \item Add the edge $(v^m_{i},t)$ to close the chain $P(m)$ to create a tour in $G^*$ (See Figure~\ref{fig:G-Star-SEE} for an SEE-tour in the $G^*$ graph of Figure~\ref{fig:G-Star}).
\end{enumerate}

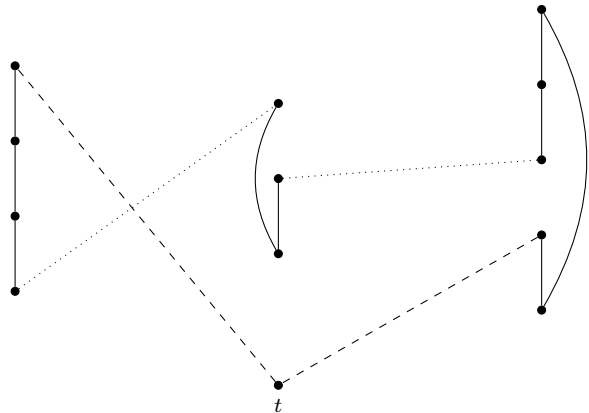
\begin{figure}[H]
    \centering
 \begin{tikzpicture}[scale=0.6]
\foreach \x[count=\xi] in {0.5,1.5,...,4}{
\node[vrtx]  (v\xi) at (2,\x) {};
}
\foreach \x[count=\xi] in {0.5,1.5,...,3}{
\node[vrtx]  (w\xi) at (5.5,.5+\x) {};
}
\foreach \x[count=\xi] in {0.5,1.5,...,5}{
\node[vrtx]  (u\xi) at (9,-.25+\x) {};
}
\draw (v1) -- (v2) -- (v3) -- (v4);
\draw (w1) -- (w2);
\draw[bend right] (w3) to (w1);
\draw (u1) -- (u2)  (u3) -- (u4) -- (u5);
\draw[bend left] (u5) to (u1);
\node[vrtx, label=below:{\footnotesize $t$}]  (t) at (5.5,-0.75) {};
\draw[dashed] (t) -- (v4);
\draw[dotted] (v1) -- (w3);
\draw[dotted] (w2) -- (u3);
\draw[dashed] (u2) -- (t);
\end{tikzpicture}
	\caption{An SEE-tour in the graph $G^*$ given in Figure \ref{fig:G-Star}.}
	   \label{fig:G-Star-SEE}
\end{figure}

Let $F(SEE)$ be the collection of SEE-tours in $G^*$. As indicated in~\cite{glover1997travelling}, $|F(SEE)|=2^m\prod_{k=1}^m|V^k|$. If $|V^k|=3$ for all $k$, then $|F(SEE)|=6^{(n-1)/3}\approx (1.817)^{n-1}$. If $|V^k|=4$ for all $k$, then $|F(SEE)|=8^{(n-1)/4}\approx (1.68)^{n-1}$. Thus finding \myred{a} best TSP tour in $F(SEE)$ is a non-trivial task. Glover and Punnen~\cite{glover1997travelling} proposed an $O(n)$ algorithm to solve the linear TSP when restricted to SEE-tours on $G^*$.


In the definition of QTSP, if the set of feasible solutions is restricted to the class of SEE-tours in $G^*$, we have an instance of QTSP-SEE. Although the linear TSP over SEE-tours can be solved in $O(n)$ time, QTSP-SEE is a much more difficult problem.

Before discussing our complexity results, we present the definition of two well-known NP-hard problems that are used in our reductions; the quadratic unconstrained binary optimization problem (QUBO) and the partition problem ({\sc PARTITION}). \pu{
\begin{description}
\item[QUBO] Given an $n\times n$ cost matrix $Q=(q_{ij})_{n\times n}$, the problem QUBO is to find  an $\vect{x}\in \{0,1\}^n$ such that $\vect{x}^TQ\vect{x}$ is minimized.
\item[PARTITION]  Given $n$ numbers $\alpha_1,\alpha_2,\ldots, \alpha_n$, the PARTITION problem is to determine if there exists subsets $S_1$ and $S_2$ of $\{1,2,\ldots ,n\}$ such that $S_1\cup S_2 =\{1,2,\ldots ,n\}$, $S_1\cap S_2 =\emptyset$, and  $\sum_{j\in S_1}\alpha_j = \sum_{j\in S_2}\alpha_j$.
\end{description}
}
\begin{theorem}\label{thm:DQTSP-SEE}
    QTSP-SEE is strongly NP-hard.
\end{theorem}
\begin{proof}
    We reduce QUBO to QTSP-SEE. From an instance of QUBO, we  construct an instance of QTSP-SEE as follows. For each variable $x_i$, $1\leq i \leq n$, of QUBO, create a 3-cycle $C(i)$. Choose an edge from each $C(i)$ and label it $i$. Now construct the graph $G^*$ using these cycles. Arbitrarily label the remaining unlabeled edges of $G^*$ as $n+1,n+2,\ldots ,m$. (See Figure 3) Consider an $m\times m$ matrix $Q^{\prime}=(q^{\prime}_{ij})_{m\times m}$ where
    \[
    q^{\prime}_{ij} = \begin{cases}
q_{ij}, & \mbox{ if } 1\leq i,j \leq n \\
0, & \mbox{ otherwise.}
\end{cases}
    \]
    Thus, \myred{$Q^{\prime}=\left[
\begin{array}{c|c}
Q & \mathcal{O}\\ \hline
\mathcal{O}^T & \mathcal{O}'
\end{array}\right]
$where $\mathcal{O}$, and $\mathcal{O}'$ are the zero matrices of size $n \times (m-n)$ and $(m-n)\times (m-n)$, respectively.}
Given any solution $\vect{x}=(x_1,x_2,\ldots ,x_n)$ of QUBO, we can construct an SEE-tour, $\tau$, in $G^*$ containing the edge $i$ if $x_i =1$ and not containing $i$ if $x_i=0$, for $1\leq i \leq n$. Note that $\tau$ contains other edges as well. It can be verified that the cost of $\tau$ with cost matrix $Q^{\prime}$ is precisely $\vect{x}^TQ\vect{x}$.

Conversely, given any SEE-tour $\tau$ in the $G^*$ obtained above, construct a vector $\vect{x}= (x_1,x_2,\ldots ,x_n)$ by assigning $x_i=1$ if and only if edge $i$ is in $\tau$, for $1\leq i \leq n$. The cost of the tour $\tau$ with cost matrix $Q^{\prime}$ is precisely $\vect{x}^TQ\vect{x}$. Since QUBO is strongly NP-hard, the result follows.
\end{proof}
\begin{figure}[H] \centering
\begin{tikzpicture}
\foreach \x[count=\xi] in {0.5,1.5,...,3}{
\node[vrtx]  (v\xi) at (2,\x) {};
}
\node[fit=(v3) (v2) (v1),ellipse,draw=gray!50,minimum width=2cm, label=above:{$C(1)$}] {};
\foreach \x[count=\xi] in {0.5,1.5,...,3}{
\node[vrtx]  (w\xi) at (4.5,\x) {};
}
\node[fit=(w3) (w2) (w1),ellipse,draw=gray!50,minimum width=2cm, label=above:{$C(2)$}] {};
\foreach \x[count=\xi] in {0.5,1.5,...,3}{
\node[vrtx]  (u\xi) at (7,\x) {};
}
\node[fit=(u1) (u2) (u3),ellipse,draw=gray!50,minimum width=2cm, label=above:{$C(3)$}] {};
\foreach \x[count=\xi] in {0.5,1.5,...,3}{
\node[vrtx]  (z\xi) at (11,\x) {};
}
\node[fit=(z1) (z2) (z3),ellipse,draw=gray!50,minimum width=2cm, label=above:{$C(n)$}] {};
\draw (v1) -- (v2) -- node[right,xshift=-0.05cm] {\footnotesize$1$} ++ (v3);
\draw[bend right] (v3) to (v1);
\draw (w1) -- (w2) -- node[right,xshift=-0.05cm] {\footnotesize$2$} ++ (w3);
\draw[bend right] (w3) to (w1);
\draw (u1) -- (u2) -- node[right,xshift=-0.05cm] {\footnotesize$3$} ++ (u3);
\draw[bend right] (u3) to (u1);
\draw (z1) -- (z2) -- node[left,xshift=0.05cm] {\footnotesize$n$} ++ (z3);
\draw[bend left] (z3) to (z1);
\foreach \x in {1,2,3}{
\foreach \y in {1,2,3}{
\draw[dotted] (v\x) -- (w\y);
\draw[dotted] (w\x) -- (u\y);
}
}
\node[vrtx, label=below:{\footnotesize $t$}]  (t) at (7,-1.0) {};
\foreach \y in {1,2,3}{
\draw[dashed] (t) -- (z\y);
\draw[dashed] (t) -- (v\y);
\path (u\y) -- node[auto=false]{\ldots} (z\y);
}
\end{tikzpicture}
	\caption{Construction of the graph $G^*$ used in the proof of Theorem~\ref{thm:DQTSP-SEE}.}
	   \label{fig:G-Star2}
\end{figure}
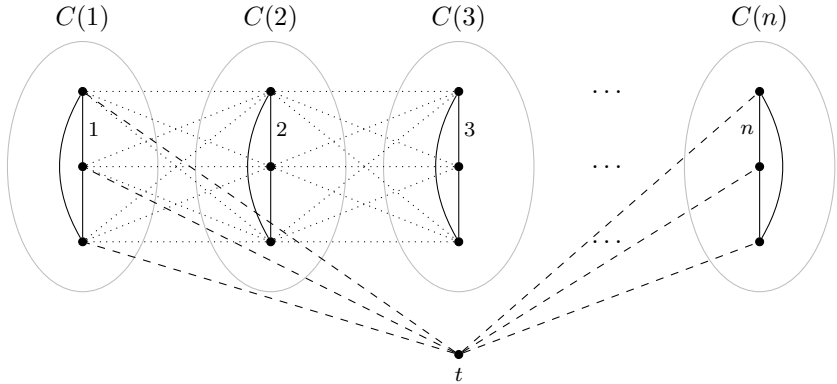

Let us now examine the complexity of some special cases of QTSP-SEE. In the definition of QTSP(p,c), if we restrict the solution set to SEE-tours in $G^*$, we have the instance QTSP(p,c)-SEE. \myred{That is,} QTSP(p,c)-SEE is precisely the special case of QTSP-SEE where the rank of the associated cost matrix is $p$ and a linear cost function is added to the quadratic costs. If the linear part is zero (i.e. homogeneous case), we denote the corresponding instance by QTSP(p,H)-SEE. Recall that \myred{there exists some vectors $\vect{a}^r$ and $\vect{b}^r$ for $r=1,\ldots,p$ such that} QTSP(p,c)-SEE can be stated as
\begin{align*}
\text{Minimize     }  q(\tau) =& \sum_{r=1}^p\left[\left(\sum_{e\in \tau} a_e^r\right)\left(\sum_{e\in \tau}b_e^r\right)\right] + \sum_{e\in \tau}c(e) \\
\text{Subject to } & \tau \in F(SEE).
\end{align*}

\begin{theorem}\label{thm:QTSP1-SEE}
 QTSP(p,c)-SEE is NP-hard even if $p=1$ and $c(e) =0$ for all \myred{$e$}.
\end{theorem}
\begin{proof}
We reduce the PARTITION problem to QTSP(1,H)-SEE.  From an instance of PARTITION \myred{with  data $\alpha_1,\ldots,\alpha_n$}, we construct an instance of QTSP(1,H)-SEE as follows.

For each $k=1,2,\ldots ,n$, create a 3-cycle $C(k)$ on the vertex set $\{v_u^k,v_y^k,v_w^k\}$. Build the graph $G^*=(V,E)$ using these cycles. Define weight for each edge $(i,j)\in E$ as follows:  For $k=1,2,\ldots ,n$, assign weight $\alpha_k$ to edge $(v_y^k,v_u^k)$ and  $-\alpha_k$ to the edge $(v_y^k,v_w^k)$. For $k=1,2,\ldots ,n-1$ assign weights $-M$ for $(v_u^k,v_y^{k+1})$ and $(v_w^k,v_y^{k+1})$ where $M= 1+\sum_{k=1}^n|\alpha_k|$. The weight of edge $(t,v_y^1)$ is $nM$, the weights of edges $(t,v^1_u)$ and $(t,v^1_w)$ are $nM + 1$, and the weights of edges $(v_u^n,t)$ and $(v_w^n,t)$ are $-M$, where $t$ is the tip vertex of $G^*$. All other edges have weight zero. (See Figure 4 for a sample $G^*$ graph constructed.) Let $a_{ij}$ denote the weight of edge $(i,j)$ constructed above and choose another set of weights, $b_{ij}$ for edge $(i,j)$, $i,j\in V$ such that $b_{ij}=a_{ij}$. Then, the objective function of QTSP(1,H)-SEE on the $G^*$ constructed above becomes $\left(\sum_{(i,j)\in \tau}a_{ij}\right )^2$ where $\tau$ is an SEE-tour in this $G^*$. Note that zero is a lower bound on the optimal objective function value of QTSP(1,H)-SEE constructed above. It can be verified that the optimal objective function value of this QTSP(1,H)-SEE is zero precisely when the required partition exists. The result follows from the NP-completeness of PARTITION~\cite{Karp1972}.\end{proof}
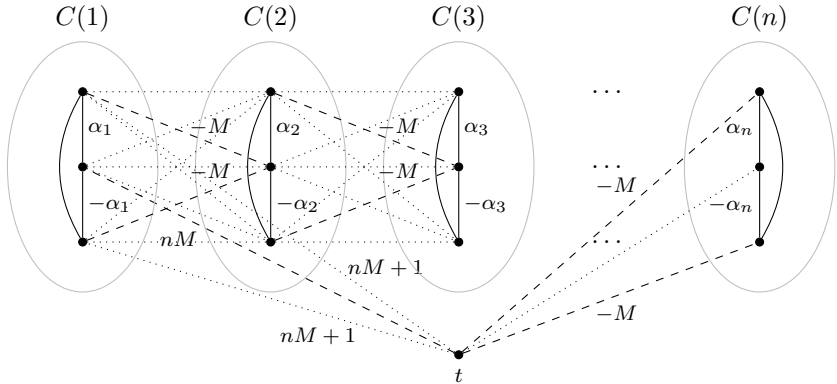
\begin{figure}[H]
    \centering
 \begin{tikzpicture}
\foreach \x[count=\xi]/\y in {0.5/w,1.5/y,2.5/u}{
\node[vrtx] (v\xi) at (2,\x){};
}
\draw[dotted] (t) -- node[right,xshift=0cm,yshift=-0.5cm] {\footnotesize$nM+1$} ++ (v1);
\draw[dashed] (t) -- node[right,xshift=-1.6cm,yshift=0.3cm] {\footnotesize$nM$} ++ (v2);
\draw[dotted] (t) -- node[right,xshift=0.9cm,yshift=-0.6cm] {\footnotesize$nM+1$} ++ (v3);
\node[fit=(v1) (v2) (v3),ellipse,draw=gray!50,minimum width=2cm, label=above:{$C(1)$}] {};
\foreach \x[count=\xi]/\y in {0.5/w,1.5/y,2.5/u}{
\node[vrtx]  (w\xi) at (4.5,\x){};
}
\node[fit=(w1) (w2) (w3),ellipse,draw=gray!50,minimum width=2cm, label=above:{$C(2)$}] {};
\foreach \x[count=\xi]/\y in {0.5/w,1.5/y,2.5/u}{
\node[vrtx]  (u\xi) at (7,\x){};
}
\node[fit=(u1) (u2) (u3),ellipse,draw=gray!50,minimum width=2cm, label=above:{$C(3)$}] {};
\foreach \x[count=\xi]/\y in {0.5/w,1.5/y,2.5/u}{
\node[vrtx]  (z\xi) at (11,\x){};
}
\node[fit=(z1) (z2) (z3),ellipse,draw=gray!50,minimum width=2cm, label=above:{$C(n)$}] {};
\foreach \x[count=\xi] in {v,w,u}{
\draw (\x1) -- node[right,xshift=-0.05cm] {\footnotesize$-\alpha_\xi$} ++ (\x2);
\draw (\x2) -- node[right,xshift=-0.05cm] {\footnotesize$\alpha_\xi$} ++ (\x3);
\draw[bend right] (\x3) to (\x1);
}
\draw (z1) -- node[left,xshift=0.05cm] {\footnotesize$-\alpha_n$} ++ (z2);
\draw (z2) -- node[left,xshift=0.05cm] {\footnotesize$\alpha_n$} ++ (z3);
\draw[bend left] (z3) to (z1);
\foreach \x in {1,2,3}{
\foreach \y in {1,3}{
\draw[dotted] (v\x) -- (w\y);
\draw[dotted] (w\x) -- (u\y);
}
}
\draw[dashed] (v1) -- node[above,xshift=0.45cm,yshift=0.2cm] {\footnotesize$-M$} ++ (w2);
\draw[dashed] (w1) -- node[above,xshift=0.45cm,yshift=0.2cm] {\footnotesize$-M$} ++ (u2);
\draw[dotted] (v2) -- (w2);
\draw[dotted] (w2) -- (u2);
\draw[dashed] (v3) -- node[above,xshift=0.45cm,yshift=-0.2cm] {\footnotesize$-M$} ++ (w2);
\draw[dashed] (w3) -- node[above,xshift=0.45cm,yshift=-0.2cm] {\footnotesize$-M$} ++ (u2);
\node[vrtx, label=below:{\footnotesize $t$}]  (t) at (7,-1.0) {};
\foreach \y in {1,2,3}{
\path (u\y) -- node[auto=false]{\ldots} (z\y);
}
\draw[dashed] (t) -- node[right,xshift=-0.3cm,yshift=-0.2cm] {\footnotesize$-M$} ++ (z1);
\draw[dotted] (t) -- (z2);
\draw[dashed] (t) -- node[right,xshift=-0.3cm,yshift=0.5cm] {\footnotesize$-M$} ++(z3);

\end{tikzpicture}
	\caption{Construction of the graph $G^*$ used in the proof of Theorem~\ref{thm:QTSP1-SEE}.  Note that the dotted edges do not belong to any optimal tour.}
	   \label{fig:G-Star3}
\end{figure}

Despite this negative result, we now show that when $p$ is fixed,  QTSP(p,H)-SEE can be solved in pseudopolynomial time and when the edge weights are non-negative it also admits an FPTAS. Recall that an instance of QTSP(p,H)-SEE is given by $p$ pairs of costs $a^r_{ij}, b^r_{ij}$ for $r=1,2,\ldots ,p$, for each edge $(i,j)\in E$. We formulate QTSP(p,H)-SEE as a {\it rank $p$ quadratic shortest path problem} (QSPP(p,H)) on a directed acyclic graph.  \myred{It is well-known that QSPP(p,H) on an acyclic digraph is NP-hard even if $p=1$~\cite{Punnen01}.}  The definition of this specific quadratic shortest path problem is given below.

Given the graph $G^*$, construct the acyclic digraph $G'=(V',E')$ \myred{with edge weight vectors $\alpha^r$, $\beta^r$, for $r=1,\ldots,p$,} as follows.  Note that the vertex set $V^k$ of cycle $C(k)$ in $G^*$ is represented by $V^k =\{v^k_1,v^k_2,\ldots ,v^k_{r_k}\}$.  Also, the edge set of $C(k)$ is $E(k) = \{e^k_1,e^k_2,\ldots ,e^k_{r_k}\}$ where $e^k_i = (v^k_i,v^k_{i+1})$ and the indices are taken modulo $r_k$.  For $k=1,2,\ldots ,m$, create $\hat{V}^k =\{\hat{v}^{k}_1, \hat{v}^{k}_2,\ldots ,\hat{v}^{k}_{r_k}\}$.  $\hat{V}^k$ can be viewed as a copy of $V^k$.  Define $V^{\prime}= \{s,t\}\cup \cup_{k=1}^{m} (V^k\cup\hat{V}^k)$. For each edge $(v^k_i,v^k_{i+1})$ in $C(k)$, introduce a directed edge $(v^k_i,\hat{v}^k_{i+1})$ and another directed edge $(v^k_{i+1},\hat{v}^k_i)$ \myred{in $E'$,} where the indices are taken modulo $r_k$. The edge $(v^k_i,\hat{v}^k_{i+1})$ represents the event of ejecting edge $e^k_i$ from $C(k)$ where a Hamiltonian cycle ``enters" $C(k)$ through $v^k_i$ and ``leaves" $C(k)$ through $v^k_{i+1}$. For each $i=1,2,\ldots ,r_k$ and each $h=1,2,\ldots ,p$, \myred{set} $\alpha^h_{v^k_i,\hat{v}^k_{i+1}} = C(a^h,k)-a^h_{e^k_i}$ and $\beta^h_{v^k_i,\hat{v}^k_{i+1}} = C(b^h,k)-b^h_{e^k_i}$, where $C(a^h,k)= \sum_{e\in C(k)}a^h_{e}$ and $C(b^h,k)=\sum_{e\in C(k)}b^h_{e}$. Similarly, the edge $(v^k_{i+1},\hat{v}^k_i)$ corresponds to ejecting edge $e^k_i$ from $C(k)$ and a Hamiltonian cycle enters $C(k)$ from $v^k_{i+1}$, traverses $v^k_{i+1},\ldots,v^k_i$, and leaves $C(k)$ through $v^k_i$.
For $h=1,2,\ldots ,p$, set $\alpha^h_{v^k_{i+1},\hat{v}^k_i} =\alpha^h_{v^k_i,\hat{v}^k_{i+1}}$ and $\beta^h_{v^k_{i+1},\hat{v}^k_i} =\beta^h_{v^k_i,\hat{v}^k_{i+1}}$. For each edge $(v^k_i, v^{k+1}_j)$ connecting vertices in $V^k$ and $V^{k+1}$ introduce a directed edge $(\hat{v}^k_i,v^{k+1}_j)$ \myred{in $E'$}. For $h=1,2\ldots ,p$, set $\alpha^h_{\hat{v}^k_i,v^{k+1}_j}=a^h_{v^k_i, v^{k+1}_j}$ and $\beta^h_{\hat{v}^k_i,v^{k+1}_j}=b^h_{v^k_i, v^{k+1}_j}$. The tip vertex $s$ is connected to $v^1_i$, for $i=1,2,\ldots,r_1$, and set the weights for edges \myred{$(s,v^1_i)$} in $E'$ as $\alpha^h_{s,v^1_i}=a^h_{t,v^1_i}$ and $\beta^h_{s,v^1_i}=b^h_{t,v^1_i}$\myred{, and corresponds to entering $C(1)$ via the edge $(t,v^1_i)$}. Similarly, \myred{every vertex $v^m_i$ is connected to $t$, for $i=1,2,\ldots,r_m$, with weights} \myred{set to} $\alpha^h_{\hat{v}^m_i,t} = a^h_{i,t}$ and $\beta^h_{\hat{v}^m_i,t} = b^h_{v^m_i,t}$\myred{, and corresponds to leaving $C(k)$ via the edge $(v^m,t)$}. The graph $G'$ constructed from the $G^*$ in Figure \ref{fig:G-Star} is shown in Figure~\ref{fig:G'-SEE}.\\
\begin{figure}[ht]
    \centering
 \begin{tikzpicture}[>=stealth, decoration={markings, mark=at position 0.25 with {\arrow{>}}}]
\foreach \x[count=\xi] in {0.5,1.5,...,4}{
  \node[vrtx]  (e\xi) at (0,\x) {};
  \node[vrtx]  (f\xi) at (2.25,\x) {};
}
\node[fit=(e4) (e3) (e2) (e1),ellipse,draw=gray!50,minimum width=2cm, label=above:{$V(1)$}] {};
\node[fit=(f4) (f3) (f2) (f1),ellipse,draw=gray!50,minimum width=2cm, label=above:{$\hat{V}(1)$}] {};
\foreach \x[count=\xi] in {0.5,1.5,...,3}{
  \node[vrtx]  (g\xi) at (4.5,.5+\x) {};
  \node[vrtx]  (h\xi) at (6.75,.5+\x) {};
}
\node[fit=(g3) (g2) (g1),ellipse,draw=gray!50,minimum width=2cm, label=above:{$V(2)$}] {};
\node[fit=(h3) (h2) (h1),ellipse,draw=gray!50,minimum width=2cm, label=above:{$\hat{V}(2)$}] {};
\foreach \x[count=\xi] in {0.5,1.5,...,5}{
  \node[vrtx]  (i\xi) at (9,-.25+\x) {};
  \node[vrtx]  (j\xi) at (11.25,-.25+\x) {};
}
\node[fit=(i1) (i2) (i3) (i4) (i5),ellipse,draw=gray!50,minimum width=2cm, label=above:{$V(3)$}] {};
\node[fit=(j1) (j2) (j3) (j4) (j5),ellipse,draw=gray!50,minimum width=2cm, label=above:{$\hat{V}(3)$}] {};
\draw[dashed,postaction={decorate}] (e1) -- (f2);
\draw[dashed,postaction={decorate}] (e2) -- (f3);
\draw[dashed,postaction={decorate}] (e3) -- (f4);
\draw[dashed,postaction={decorate}] (e4) -- (f1);
\draw[dashed,postaction={decorate}] (e1) -- (f4);
\draw[dashed,postaction={decorate}] (e2) -- (f1);
\draw[dashed,postaction={decorate}] (e3) -- (f2);
\draw[dashed,postaction={decorate}] (e4) -- (f3);
\foreach \x in {1,2,3,4}{
  \foreach \y in {1,2,3}{
    \draw[dotted,postaction={decorate}] (f\x) to (g\y);
  }
}
\foreach \x in {1,2,3}{
  \foreach \y in {1,2,3}{
    \ifnum\x=\y [\else \draw[dashed,postaction={decorate}] (g\x) -- (h\y);] \fi
  }
}
\foreach \x in {1,2,3}{
  \foreach \y in {1,2,3,4,5}{
    \draw[dotted,postaction={decorate}] (h\x) to (i\y);
  }
}
\draw[dashed,postaction={decorate}] (i1) -- (j2);
\draw[dashed,postaction={decorate}] (i2) -- (j3);
\draw[dashed,postaction={decorate}] (i3) -- (j4);
\draw[dashed,postaction={decorate}] (i4) -- (j5);
\draw[dashed,postaction={decorate}] (i5) -- (j1);
\draw[dashed,postaction={decorate}] (i1) -- (j5);
\draw[dashed,postaction={decorate}] (i2) -- (j1);
\draw[dashed,postaction={decorate}] (i3) -- (j2);
\draw[dashed,postaction={decorate}] (i4) -- (j3);
\draw[dashed,postaction={decorate}] (i5) -- (j4);
\node[vrtx, label=left:{\footnotesize $s$}]  (t1) at (-1.5,2) {};
\node[vrtx, label=right:{\footnotesize $t$}]  (t2) at (12.75,2.25) {};
\foreach \y in {1,2,3,4}{
  \draw[dashed,postaction={decorate}] (t1) -- (e\y);
}
\foreach \y in {1,2,3,4,5}{
  \draw[dashed,postaction={decorate}] (j\y) -- (t2);
}
\end{tikzpicture}
	\caption{$G'$ constructed from the graph $G^*$ given in Figure \ref{fig:G-Star}.}
	\label{fig:G'-SEE}
\end{figure}
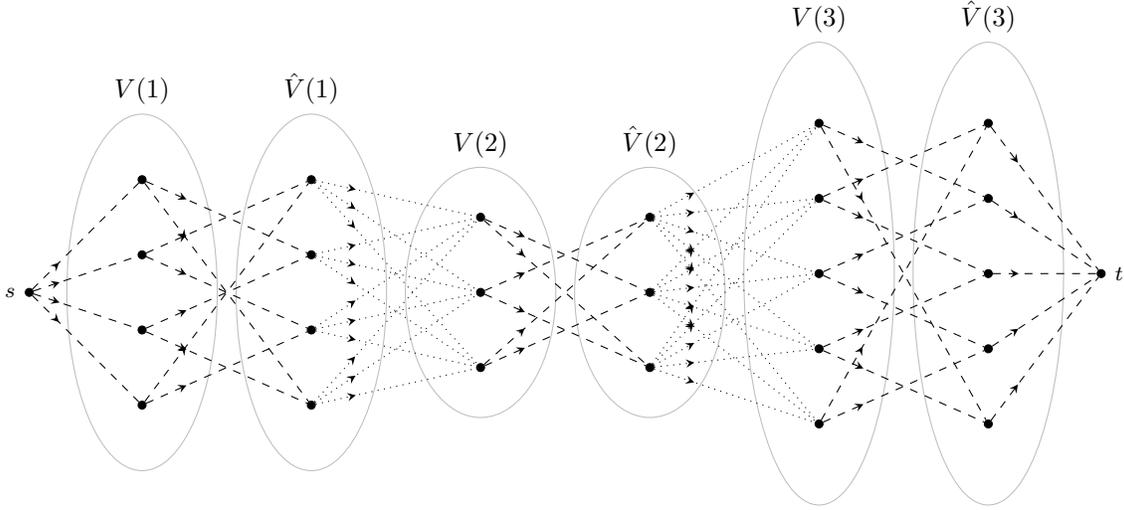

Consider the homogeneous rank $p$ quadratic shortest path problem on $G'$,
\begin{eqnarray*}
    QSPP(p,H,G'): &\textrm{Minimize } & q(P) = \sum_{r=1}^p \left(\sum_{e\in P}a^r_e\right) \left( \sum_{e\in P}b^r_e\right) \\
    & \textrm{Subject to } & P \in \mathcal{P}_{s,t},
\end{eqnarray*}
where $\mathcal{P}_{s,t}$ is the set of all $s-t$ paths in $G'$.
\begin{theorem}\label{pth1}
	From an optimal ($\epsilon$-optimal) solution of QSPP(p,H,G$'$), an optimal ($\epsilon$-optimal) solution to QTSP(p,H)-SEE can be recovered in linear time.
\end{theorem}
\begin{proof}
	From the construction of $G'$, it can be verified that there is a one-to-one correspondence between SEE-tours in $G^*$ and $s-t$ paths in $G'$.  Moreover, the objective function values of the corresponding solutions of QTSP(p,H)-SEE and QSPP(p,H,G$'$) are identical, and the result follows.
\end{proof}

\pu{In view of Theorem~\ref{pth1}, to solve QTSP-SEE(p,H) (either by an exact algorithm or by an approximation algorithm), it is enough to solve a quadratic shortest path problem with cost matrix of rank $p$ on an acyclic digraph (QSPP(p,H)). In QSPP(p,H), we assume that the cost matrix is given in rank decomposition form as vectors $\vect{a}^r$ and $\vect{b}^r$ for $r=1,2,\ldots,p$. For the computational complexity of the quadratic shortest path problem and its various special cases, we refer to~\cite{rostami2018quadratic, hu2018special}.
}

\subsection{The QSPP(p,H)}

We now present a labelling algorithm to solve QSPP(p,H)  in pseudopolynomial time on an acyclic \myred{multi}digraph $G=(V,E)$. Construct a distance function $\vect{\delta}:(E,r)\rightarrow \mathbb{R}^{2p}$, which stores the values for $\vect{a}^r$ and $\vect{b}^r$ for $r=1,2,\ldots,p$.  Our algorithm stores a collection of distance label vectors, denoted $\Omega_j$, on each $j\in V$.  A label $\vect{d}$ on vertex $j$ represents the existence of a unique path $P^d_j$ from $s$ to $j$ such that the sums of the costs of $\vect{a}^r$ and $\vect{b}^r$ of edges in $P^d_j$, for each $r=1,2,\ldots,p$, are equal to the entries of $\vect{d}$.  Then, to solve QSPP(p,H), it suffices to find the distance label at $t$ which minimizes the sum of products of its stored values.  We now give the details of our approach.

Note that only vertices of $G$ that lie on some $s-t$ path in $G$ are relevant to QSPP(p,H).  Thus, we can remove all vertices of $G$ that are not reachable from $s$ and those from which $t$ is not reachable.  Such vertices can be identified in $O(|V|+|E|)$ time by two applications of breadth-first search.  Thus, without loss of generality, we assume that each vertex of $G$ lies on some $s-t$ path in $G$, no incoming arcs to $s$, no outgoing arcs from $t$, the vertex set $V=\{1,2,\ldots,n\}$ and the vertex labels follow topological order, $s=1$ and $t=n$.

For each \myred{arc} $(i,j)\in E$, let $\vect{\delta}_{ij}\in \mathbb{R}^{2p}$ be defined as
\begin{equation*}
    \delta_{ij}(h) =
    \begin{cases}
        a^h_{ij} & \text{ if } h=1,2,\ldots,p \\
        b^{h-p}_{ij} & \text{ if } h=p+1,p+2,\ldots, 2p.
    \end{cases}
\end{equation*}

Our pseudopolynomial algorithm to solve QSPP(p,H) maintains a collection $\Omega_j$, of distance label vectors \myred{for all}  $j\in V$.  Each \myred{distance label} vector $\vect{d}\in \Omega_j$ belongs to $\mathbb{R}^{2p}$ and represents a unique path $P^{\vect{d}}_j$ from $1$ to $j$ in $G$ such that
\begin{equation*}
    d(h) =
    \begin{cases}
        \sum_{e\in P^{\vect{d}_j}} a^{r}_e & \text{ if } h=1,2,\ldots,p\\
        \sum_{e\in P^{\vect{d}_j}} b^{r-p}_e & \text{ if } h=p+1,p+2,\ldots,2p.
    \end{cases}
\end{equation*}
For each $j\in V$, let $I(j)=\{i:(i,j)\in E\}$.  Then, given $\Omega_i$ for $i\in I(j)$, the set $\Omega_j$ can be constructed by choosing distinct elements of the multiset
\begin{equation}\label{set1}
    \{\vect{d}+\vect{\delta}_{ij}:\vect{d}\in \Omega_i,i\in I(j)\}.
\end{equation}
Starting with $\Omega_1$ consisting of the zero vector in $\mathbb{R}^{2p}$, the sets $\Omega_1,\Omega_2,\ldots,\Omega_n$ can be generated using the \myred{formula}~(\ref{set1}).  Let $\vect{d}^*\in \Omega_n$ be such that
\[\sum_{i=1}^p d^*(i) d^*(p+i) = \min_{\vect{d}\in \Omega_n} \left\{ \sum_{i=1}^{p} d(i)d(p+i)\right\} .\]
Then $\sum_{i=1}^p d^*(i)d^*(p+i)$ gives the optimal objective function value of QSPP(p,H) on $G$ with $s=1$, $t=n$ and each vertex in $G$ lies on some path from $1$ to $n$ in $G$.  The validity of this follows from the recursion defined by~(\ref{set1}).  Note that each distance label vector $\vect{d}\in \Omega_j$ is such that $\vect{d}=\vect{u}+\vect{\delta}_{ij}$ for some $i\in I(j)$ and $\vect{u}\in \Omega_i$.  For each distance label $\vect{d}\in \Omega_j$, we maintain $pred(\vect{d})=i$  \myred{which stores the predecessor vertex of distance label $\vect{d}$}, and $pointer(\vect{d})$ which is a pointer to \myred{an appropriate distance label} in $\Omega_i$.  A formal description of the algorithm is given below.

\begin{algorithm}[H]
\caption{fixed-rank QSPP}
\begin{algorithmic}
\Require Directed acyclic \myred{multi}graph $G=(V,E)$ with costs $a_e^r$, $b_e^r$ for $r=1,\ldots,p$ and $e\in E$, specified vertices $s$ and $t$.
\Ensure \myred{A} shortest path from $s$ to $t$.
\State Remove each vertex not reachable from $s$ and each vertex from which $t$ cannot be reached
\State Label vertices in topological order with $s=1$ and $t=n$.
\State Construct distance function $\vect{\delta}$ from $\vect{a}^r$, $\vect{b}^r$ for $r=1,\ldots,p$
\State $\Omega_1 = \myred{\{}\vect{0} \myred{\}} \in \mathbb{R}^{2p}$ \myred{\Comment{The set of distance labels at vertex 1 contains only the zero vector.}}
\For{$j=2,3,\ldots,n$}
    \State $\bar{\Omega}=\emptyset$ \myred{\Comment{Begin with the empty set.}}
    \For{$i\in I(j)$}
        \For{$\myred{\vect{w}}\in \Omega_{i}$}
            \State $\vect{d} = \vect{w} + \vect{\delta}_{ij}$ \myred{\Comment{Compute distance label vector.}}
            \State $pred(\vect{d})=i$ \myred{\Comment{Store the predecessor vertex of $\vect{d}$.}}
            \State $pointer(\vect{d})=\myred{\vect{w}}$ \myred{\Comment{Store a pointer to the predecessor distance label.}}
            \State $\bar{\Omega}=\bar{\Omega}\cup \{\vect{d}\}$ \myred{\Comment{Add distance label to set.}}
        \EndFor
    \EndFor
    \State $\Omega_j=$ distinct elements of $\bar{\Omega}$ \myred{\Comment{Remove duplicate vectors.}}
\EndFor
\State Choose $\vect{u}\in \Omega(n)$ such that $\sum_{i=1}^p u(i) u(p+i) = \min_{\vect{d}\in \Omega_n} \left\{\sum_{i=1}^{p} d(i)d(p+i)\right\}$.
\State Trace the path \myred{$P^{\vect{u}}_n$} determined by \myred{$\vect{u}$} using $pred(\myred{\vect{u}})$ and $pointer(\myred{\vect{u}})$.
\newline \Return \myred{$P^{\vect{u}}_n$}
\end{algorithmic}
\end{algorithm}

\begin{lemma}
    $|\Omega_n|\geq |\Omega_j|$ for $j=1,2,\ldots,n$.
\end{lemma}
\begin{proof}
    Let $P=(\pi(1),\pi(2),\ldots,\pi(r))$ be any path from vertex $1$ to $n$ in $G$.  Consider a vertex $\pi(i)$, $i\in \{1,2,\ldots,r-1\}$.  Since the elements of $\Omega_{\pi(i)}$ are distinct vectors, the vectors that belong to $\{\vect{d}+\vect{\delta}_{\pi(i)\pi(i+1)}:\vect{d}\in \Omega_{\pi(i)}\}$ are distinct. Thus, $|\Omega_{\pi(i+1)}|\geq |\Omega_{\pi(i)}|$.  Since each vertex in $G$ belongs to some path joining vertex $1$ to vertex $n$, the result follows.
\end{proof}

\begin{theorem} \label{thm:QSPP(p,H)}
    QSPP(p,H) can be solved on an acyclic \myred{multi}digraph $G$ in $O(mn^{2p+1}U)$ time, where \\$U=\prod_{h=1}^p \max_{e\in E} |a^h_e| \max_{e\in E} |b^h_e|$, for any fixed $p$.
\end{theorem}
\begin{proof}
	A topological order of the vertices in \myred{multi}digraph $G$ can be obtained in $O(n+m)$ time.  For each $h=1,2,\ldots,p$, the number of possible distinct values of $a^h$ for a label at any vertex is bounded by $2(n-1)\cdot \max_e |a^h_e|$.  Similarly, the number of distinct values for $b^h$ is bounded by $2(n-1)\cdot \max_e |b^h_e|$.  That is, $|\Omega_j| \leq n^{2p}U$ for any $j\in V$, where $U=\prod_{h=1}^p \max_e |a^h_e| \max_e |b^h_e|$.  To generate each $\Omega_j$, we consider each $\Omega_i$ such that $i\in I(j)$, and $|\bar{\Omega}|\leq mn^{2p}U$.  The distinct elements of $\bar{\Omega}$ can be found in $O(mn^{2p}U)$ time, and hence, all $\Omega_j$ can be constructed in $O(mn^{2p+1}U)$ time.  Selecting the minimum $\myred{\vect{u}}\in \Omega(n)$ such that $\sum_{i=1}^p u(i) u(p+i) = \min_{d\in \Omega_n} \{\sum_{i=1}^{2p} d(i)d(p+i)\}$ can be done in $O(p|\Omega_n|)$ time, and the result follows.
\end{proof}
From Theorem~\ref{thm:QSPP(p,H)}, it follows that QSPP(p,H) on an acyclic digraph can be solved in pseudopolynomial time when $p$ is fixed.  As a consequence, QTSP(p,H)-SEE can be solved in pseudopolynomial time for fixed $p$.
\begin{corollary}
 QTSP(p,H)-SEE can be solved in $O(mn^{2p+1}U)$ time, where $U=\prod_{h=1}^p \max_{e\in E} |a^h_e| \max_e |b^h_e|$.
\end{corollary}

\pu{Let us now observe that the QSPP(p,H) can be solved a s a sequence of equality type resource-constrained shortest path problems (ECSPP)~\cite{turner2012variants}. The average performance of this approach is likely to be weaker than that of the labelling algorithm when restricted to acyclic digraphs; this approach does not require the graph to be acyclic. However, the complexity of the procedure depends on that of solving ECSPP. Let $\eta_1,\eta_2,\ldots ,\eta_p$ be a set of parameters. Introduce the constraints $\sum_{e\in P}b^r_e=\eta_r$ for $r=1,2,\ldots ,p$ and consider the cost vector $\bar{\vect{a}}$ where $\bar{a}_e=\sum_{r=1}^p\eta_ra^r_e$. Solve the resulting ECSPP and let $P(\eta)$ be the resulting optimal solution. Repeating this for all possible values of $\eta=(\eta_1,\eta_2,\ldots ,\eta_p)$ and choosing the best solution amongst the solutions of these ECSPPs provides an optimal solution to QSPP(p,H). A variation of this solution approach can be extended to construct an FPTAS for QSPP(p,H) with non-negative weights.
}

We now turn our attention to establishing that QSPP(p,H) admits an FPTAS, and hence QTSP(p,H)-SEE also admits an FPTAS.
\begin{theorem}\cite{Mittal2013} \label{thm:mittal}
	For fixed $m$ and $X\subseteq\{0,1\}^n$, let $f_i:X \rightarrow \RR_+$ for $i=1,2,\ldots,m$.  Let $h:\RR^m_+\rightarrow \RR_+$ be any function that satisfies:
	\begin{enumerate}
		\item $h(\vect{y}) \leq h(\vect{y}')$ for all $\vect{y},\vect{y}'\in \RR^m_+$ such that $y_i\leq y'_i$ for all $i=1,2,\ldots,m$; and
		\item \myred{$h(\lambda \vect{y})\leq \lambda^d h(\vect{y}')$} for all $\vect{y}\in \RR^m_+$ and $\lambda > 1$ for some fixed \myred{$d>0$}.
	\end{enumerate}
	There is an FPTAS for solving the general optimization problem: Minimize $g(\vect{x}) = h(f_1(\vect{x}),f_2(\vect{x}),\ldots,f_m(\vect{x}))$, $\vect{x}\in X$ if the following exact problem can be solved in pseudopolynomial time: Given $k\in \ZZ$, $(c_1,c_2,\ldots,c_n)\in \ZZ^n_+$, does there exist $\vect{x}\in X$ such that $\sum_{i=1}^n c_ix_i=k$?
\end{theorem}

Consider the homogenous fixed-rank quadratic optimization problem (rank-QOP), with rank $p$:
\begin{eqnarray*}
	&\textrm{Minimize } & q(\vect{x}) =  \sum_{r=1}^p a^T_r \vect{x}\cdot b^T_r \vect{x} \\
	&\textrm{Subject to } & \vect{x} \in X,
\end{eqnarray*}
where $a_r,b_r\in\ZZ^n_+$ and $X\subseteq \{0,1\}^n$.  It is clear that the conditions of Theorem~\ref{thm:mittal} are satisfied with \myred{$d=2$}.  We have the following corollary.
\begin{corollary}
	There exists an FPTAS for solving (rank-QOP) if the following exact problem can be solved in pseudopolynomial time: Given $k\in \ZZ$, $(c_1,c_2,\ldots,c_n)\in \ZZ^n_+$, does there exist $\vect{x}\in X$ such that $\sum_{i=1}^n c_ix_i=k$?
\end{corollary}

Note that the exact shortest path problem is NP-hard. We relax the problem to that of finding a shortest walk that minimizes the QSPP(p,H) objective function.  An optimal solution to the relaxed problem will have the same value as the optimal solution to the original problem since removing all cycles from any $s-t$ walk gives an $s-t$ path.  Assuming that the weights are nonnegative, the exact problem can be solved in $O(nmk)$ time by dynamic programming~\cite{Mittal2013}.  We now have the following corollaries which result from this discussion and the construction given above.
\begin{corollary} \label{thm:FPTAS QSPP(p,H)}
	QSPP(p,H) and QTSP(p,H)-SEE both admit an FPTAS when $\vect{a},\vect{b}\geq \vect{0}$.
\end{corollary}

The instance of QTSP(A) when the family of tours is restricted to $F(SEE)$ is denoted by QTSP(A)-SEE. Our reduction of QTSP(p,H)-SEE to QSPP(p,H) discussed above cannot be applied directly to solve QTSP(A)-SEE. The reduction, however, can be modified to take into consideration the cost arising from adjacent pairs of edges to get an equivalent instance of adjacent QSPP(p,H) on an acyclic graph. Since the adjacent QSPP on an acyclic graph can be solved in polynomial time~\cite{rostami2015quadratic}, QTSP(A)-SEE can be solved in polynomial time. We present below a simple $O(n^2)$ algorithm to solve QTSP(A)-SEE directly.

We note that the QTSP(A) objective function only depends on consecutive edges in a tour.  Moreover, since SEE-tours visit the cycles in $G^*$ sequentially, a dynamic programming algorithm naturally emerges.  Refer to any chain $P(k)$ that may be produced in the construction of an SEE-tour in $G^*$ as an \emph{SEE-Hamiltonian path} of length $k$.  \myred{A} minimum cost SEE-Hamiltonian path of length $k$ from $t$ to $v^k_i$ can be expressed as a minimum of the SEE-Hamiltonian paths of length $k-1$ plus the costs induced by connecting each path to $v^k_i$.  We now give the details of the procedure.

Without loss of generality, assume that the input for QTSP(A)-SEE is given as cost of paths of length two in $G^*$. \myred{That is,} for any 2-path $u-v-w$ with $v$ as the middle vertex, a cost $q(u,v,w)$ is given. Note that $q(u,v,w)=q(w,v,u)$.  Let $f\left(v^k_i,v^k_{i+1}\right)$ be the length of \myred{a} smallest SEE-Hamiltonian path in $G^*$ from $t$ to $v^k_i$ when edge $\left(v^k_i,v^k_{i+1}\right)$ is ejected \myred{(assuming it contains edge $(t,v^1_{i+1})$ or $(v^{k-1}_j, v^k_{i+1})$)} and let
\[g\left(v^k_i\right) = \sum_{s=1}^{r_k}q\left(v^k_s,v^k_{s+1},v^k_{s+2}\right)-q\left (v^k_i,v^k_{i+1},v^k_{i+2}\right )-q\left (v^k_{i-1},v^k_i,v^k_{i+1}\right ).\]
In the above expression and that follows, we assume that the indices $r_k+1 \equiv 1, r_k+2 \equiv 2,$ and $ 0 \equiv r_k$. Then for $k=2,3,\ldots m$,
\begin{align*}
f(v^k_i,v^k_{i+1})  = \min_{1\leq j\leq r_{k-1}}&\left\{f\left(v^{k-1}_j,v^{k-1}_{j-1}\right ) \right . +q\left(v^{k-1}_{j+1},v^{k-1}_j,v^{k}_{i+1}\right)+q\left(v^{k-1}_j,v^{k}_{i+1},v^k_{i+2}\right)+g\left (v^k_i\right ),\\
&f\left(v^{k-1}_j,v^{k-1}_{j+1}\right )  +q\left(v^{k-1}_{j-1},v^{k-1}_j,v^{k}_{i+1}\right)+
q\left(v^{k-1}_j,v^{k}_{i+1},v^k_{i+2}\right)+\left . g\left(v^k_i\right ) \right\} ,
\end{align*}
\myred{and}
\myred{
\begin{align*}
f(v^k_i,v^k_{i-1})  = \min_{1\leq j\leq r_{k-1}}&\left\{f\left(v^{k-1}_j,v^{k-1}_{j-1}\right ) \right . +q\left(v^{k-1}_{j+1},v^{k-1}_j,v^{k}_{i-1}\right)+q\left(v^{k-1}_j,v^{k}_{i-1},v^k_{i-2}\right)+g\left (v^k_{i-1}\right ),\\
&f\left(v^{k-1}_j,v^{k-1}_{j+1}\right )  +q\left(v^{k-1}_{j-1},v^{k-1}_j,v^{k}_{i-1}\right)+
q\left(v^{k-1}_j,v^{k}_{i-1},v^k_{i-2}\right)+\left . g\left(v^k_{i-1}\right ) \right\} .
\end{align*}
}
The value\myred{s} of $f\left(v^1_i,v^1_{i+1}\right)$ \myred{and $f\left(v^1_i,v^1_{i-1}\right)$} for $1\leq i\leq r_1$ can be calculated directly to initiate the above recursion. Thus we can compute the value of the SEE-Hamiltonian path from $t$ to $v^m_i$ for each $i=1,2,\ldots ,r_m$. Adding the arc $(v^m_i,t)$ for $i=1,2,\ldots,r_m$ yields a corresponding SEE-tour and  the best such tour gives an optimal solution to QTSP(A)-SEE. The foregoing discussions can be summarized in the theorem below.
\begin{theorem}
	QTSP(A)-SEE can be solved in $O(n^2)$ time.
\end{theorem}

\section{Double edge ejection tours on $G^*$} \label{sec:DEE}

In this section we consider a special class of tours, called double edge ejection tours (DEE-tours), introduced by Glover and Punnen~\cite{glover1997travelling}.  We present various complexity results regarding QTSP and its variations restricted to this class.

The family of double edge ejection (DEE) tours in $G^*$ consists of all tours which can be obtained by the following steps.
\begin{enumerate}
    \item Begin by extending two edges $(t,v^1_i)$ and $(t,v^1_{i+1})$ from $t$ to the cycle $C(1)$ and ejecting an edge $(v^1_i,v^1_{i+1})$ from $C(1)$.  The result creates an expanded cycle $D(1)$ which includes all vertices of $C(1)$ and $t$.
    \item For each $k$ from $1$ to $m-1$, in that order, select an edge $(v^k_j,v^k_{j+1})$ of $C(k)$ where $i\neq j$, and any edge $(v^{k+1}_s,v^{k+1}_{s+1})$ of $C(k\myred{+1})$.  Eject these two edges and add either the two edges $(v^k_j,v^{k+1}_j)$ and $(v^{k}_{j+1},v^{k+1}_{s+1})$ or the two edges $(v^k_j,v^{k+1}_{s+1})$ and $(v^{k}_{j+1},v^{k+1}_{s})$, creating the expanded cycle $D(k+1)$ containing the vertices of $D(k)$ and $C(k+1)$.
    \item The cycle $D(m)$ is a DEE-tour in $G^*$ (See Figure \ref{fig:G-Star-DEE} for a DEE-tour in the $G^*$ graph of Figure \ref{fig:G-Star}).
\end{enumerate}

\begin{figure}[H]
    \centering
\begin{tikzpicture}[scale=0.6]
\foreach \x[count=\xi] in {0.5,1.5,...,4}{
\node[vrtx]  (v\xi) at (2,\x) {};
}
\foreach \x[count=\xi] in {0.5,1.5,...,3}{
\node[vrtx]  (w\xi) at (5.5,.5+\x) {};
}
\foreach \x[count=\xi] in {0.5,1.5,...,5}{
\node[vrtx]  (u\xi) at (9,-.25+\x) {};
}
\draw (v1) -- (v2)  (v3) -- (v4);
\draw[bend right] (w3) to (w1);
\draw (u1) -- (u2)  (u3) -- (u4) -- (u5);
\draw[bend left] (u5) to (u1);
\node[vrtx, label=below:{\footnotesize $t$}]  (t) at (5.5,-0.75) {};
\draw[dashed] (t) -- (v4) (t) -- (v1);
\draw[dotted] (v2) -- (w2) (v3) -- (w3);
\draw[dotted] (w2) -- (u2);
\draw[dotted] (w1) -- (u3);
\end{tikzpicture}
	\caption{A DEE-tour in the graph $G^*$ given in Figure \ref{fig:G-Star}.}
	   \label{fig:G-Star-DEE}
\end{figure}
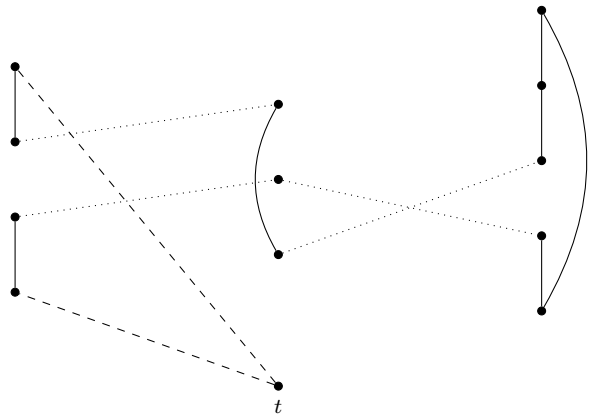

The variation of QTSP when the family of feasible solutions are restricted to DEE-tours in $G^*$ is denoted by QTSP-DEE.  Let $F(DEE)$ be the collection of all DEE-tours in $G^*$.  As indicated in \cite{glover1997travelling}, $|F(DEE)|=2^{m-1}\prod_{k=1}^m |V^k|\prod_{k=1}^{m-1}|V^{k-1}|$.  If $|V^k|=3$ for all $k$, then $|F(DEE)|=2^{m-1}3^{2m-1}\approx (1.26)^{n-4}\cdot (1.44)^{2n-7} \approx (2.61)^{n-4}$.  If $|V^k|=4$ for all $k$, then $|F(DEE)|=2^{m-1}4^{2m-1}\approx (1.19)^{n-5}\cdot 2^{n-3}$.  Despite the fact that this is an exponential class of tours, when the feasible solutions are restricted to DEE-tours in $G^*$,  the linear TSP can be solved in $O(n)$ time~\cite{glover1997travelling}. This simplicity however does not extend to QTSP-DEE.
\begin{theorem}\label{Thm:QTSP-DEE}
	QTSP-DEE is strongly NP-hard.
\end{theorem}

\begin{proof}
We reduce QUBO to QTSP-DEE.  Without loss of generality, from an instance of UBQP on $2n + 1$ variables, we construct an instance of QTSP-DEE as follows.  Create a 3-cycle $C(i)$ for $i=1,\ldots,n$.  Choose two edges from each $C(i)$, $i=1,\ldots,n-1$ and a single edge from $C(n)$ and label these $1,2,\ldots,2n+1$.  Now construct the graph $G^*$ using these cycles.  Arbitrarily label the remaining unlabelled edges of $G^*$ as $2n+2,2n+3,\ldots,m$.  Consider an $m\times m$ matrix $Q'=(q'_{ij})_{m\times m}$ where
\[
    q^{\prime}_{ij} = \begin{cases}
q_{ij}, & \mbox{ if } 1\leq i,j \leq 2n+1 \\
0, & \mbox{ otherwise.}
\end{cases}
    \]
Thus, $Q^{\prime}=\left[
\begin{array}{c|c}
Q & \mathcal{O}\\ \hline
\mathcal{O}^T & \mathcal{O}'
\end{array}\right]
$where $\mathcal{O}$, and $\mathcal{O}'$ are the zero matrices of size $2n+1 \times (m-2n-1)$ and $(m-2n-1)\times (m-2n-1)$, respectively.
Given any solution $\vect{x}=(x_1,x_2,\ldots ,x_{2n+1})$ of QUBO, we can construct an SEE-tour, $\tau$, in $G^*$ containing the edge $i$ if $x_i =1$ and not containing $i$ if $x_i=0$, for $1\leq i \leq 2n+1$. Note that $\tau$ contains other edges as well. It can be verified that the cost of $\tau$ with cost matrix $Q^{\prime}$ is precisely $\vect{x}^TQ\vect{x}$.

Conversely, given any DEE-tour $\tau$ in the $G^*$ obtained above, construct a vector $\vect{x}= (x_1,x_2,\ldots ,x_{2n+1})$ by assigning $x_i=1$ if and only if edge $i$ is in $\tau$, for $1\leq i \leq 2n+1$. The cost of the tour $\tau$ with cost matrix $Q^{\prime}$ is precisely $\vect{x}^TQ\vect{x}$. Since UBQP is strongly NP-hard, the result follows.
\end{proof}

Let us now examine the complexity of some special cases of QTSP restricted to DEE-tours.  The problem QTSP(p,H) where the family of feasible solutions is restricted to DEE-tours on $G^*$ is called double edge ejection QTSP with rank $p$, and is denoted by QTSP(p,H)-DEE. We have the analogous definition for QTSP-DEE(p,c).

\begin{theorem}\label{thm:QTSP(p,H)-DEE}
	QTSP(p,c)-DEE is NP-hard even if $p=1$ and $c(e)=0$ for all $e\in E$.
\end{theorem}
\begin{proof}
We reduce the PARTITION problem to QTSP(1,H)-DEE. From an instance of PARTITION \myred{with the given data $\alpha_1,\ldots,\alpha_n$}, construct an instance of QTSP(1,H)-DEE as follows.

For each $k=1,2,\ldots ,n$ create a 3-cycle $C(k)$ on the vertex set $\{u_k,v_k,w_k\}$. Build the graph $G^*=(V,E)$ using these cycles. Introduce a weight for each edge $(i,j)\in E$ as follows:  For $k=1,2,\ldots ,n$, assign weight $\alpha_k$ to edge $(v_k,u_k)$,  $-\alpha_k$ to the edge $(v_k,w_k)$, and $M$ to $(w_k,u_k)$, where $M= 1+n\left(\sum_{k=1}^n|\alpha_k|\right )$. The weights of the edges $(t,v_1), (t,u_1)$ and $(t,w_1)$ are $-\left(\sum_{k=1}^n\alpha_k\right )/4$.  All other edges have weight zero. Let $a_{ij}$ denote the weight of edge $(i,j)$ constructed above and choose another set of weight $b_{ij}$ which is the same as $a_{ij}$. Then the objective function of QTSP(1,H)-DEE on the $G^*$ constructed above becomes $\left(\sum_{(i,j)\in \tau}a_{ij}\right )^2$, where $\tau$ is a DEE-tour in this $G^*$. It may be noted that from each 3-cycle $C(k)$, two edges are to be ejected. In any optimal solution to the constructed instance of QTSP(1,H)-DEE, one of the ejected edge from each cycle must be the one with weight $M$. Thus for the other ejected edge, one \myred{needs} to choose an edge with weight $\alpha_k$ or $-\alpha_k$.  It can be verified that the optimal objective function value of this QTSP(1,H)-DEE is zero precisely when the required partition exists. The \myred{result} follows from the NP-completeness of PARTITION~\cite{Karp1972}.
\end{proof}
\begin{figure}[H]
    \centering
 \begin{tikzpicture}
\foreach \x[count=\xi]/\y in {0.5/w,1.5/y,2.5/u}{
\node[vrtx] (v\xi) at (2,\x){};
}
\node[vrtx, label=below:{\footnotesize $t$}]  (t) at (6,-1.0) {};
\draw[dashed] (t) -- (v1); 
\draw[dashed] (t) -- (v2); 
\draw[dashed] (t) -- (v3); 
\node[fit=(v1) (v2) (v3),ellipse,draw=gray!50,minimum width=2cm, label=above:{$C(1)$}] {};
\foreach \x[count=\xi]/\y in {0.5/w,1.5/y,2.5/u}{
\node[vrtx]  (w\xi) at (4.5,\x){};
}
\node[fit=(w1) (w2) (w3),ellipse,draw=gray!50,minimum width=2cm, label=above:{$C(2)$}] {};
\foreach \x[count=\xi]/\y in {0.5/w,1.5/y,2.5/u}{
\node[vrtx]  (u\xi) at (7,\x){};
}
\node[fit=(u1) (u2) (u3),ellipse,draw=gray!50,minimum width=2cm, label=above:{$C(3)$}] {};
\foreach \x[count=\xi]/\y in {0.5/w,1.5/y,2.5/u}{
\node[vrtx]  (z\xi) at (11,\x){};
}
\node[fit=(z1) (z2) (z3),ellipse,draw=gray!50,minimum width=2cm, label=above:{$C(n)$}] {};
\foreach \x[count=\xi] in {v,w,u}{
\draw (\x1) -- node[right,xshift=-0.05cm] {\footnotesize$-\alpha_\xi$} ++ (\x2);
\draw (\x2) -- node[right,xshift=-0.05cm] {\footnotesize$\alpha_\xi$} ++ (\x3);
\draw (\x3) to [bend right=30] node[midway,left] {\footnotesize$M$} (\x1);
}
\draw (z1) -- node[left,xshift=0.05cm] {\footnotesize$-\alpha_n$} ++ (z2);
\draw (z2) -- node[left,xshift=0.05cm] {\footnotesize$\alpha_n$} ++ (z3);
\draw (z3) to [bend left=30] node[midway,right]{\footnotesize$M$}(z1);
\foreach \x in {1,2,3}{
\foreach \y in {1,3}{
\draw[dotted] (v\x) -- (w\y);
\draw[dotted] (w\x) -- (u\y);
}
}
\draw[dotted] (v1) -- (w2);
\draw[dotted] (w1) -- (u2);
\draw[dotted] (v2) -- (w2);
\draw[dotted] (w2) -- (u2);
\draw[dotted] (v3) -- (w2);
\draw[dotted] (w3) -- (u2);
\foreach \y in {1,2,3}{
\path (u\y) -- node[auto=false]{\ldots} (z\y);
}
\draw[dotted] (t) -- (z1);
\draw[dotted] (t) -- (z2);
\draw[dotted] (t) -- (z3);

\end{tikzpicture}
	\caption{Construction of the graph $G^*$ used in the proof of Theorem~\ref{thm:QTSP(p,H)-DEE}.  Note that the dashed edges have weight $-\left(\sum_{k=1}^n\alpha_k\right )/4$.}
	\label{fig:G-Star-DEE(p,H)}
\end{figure}
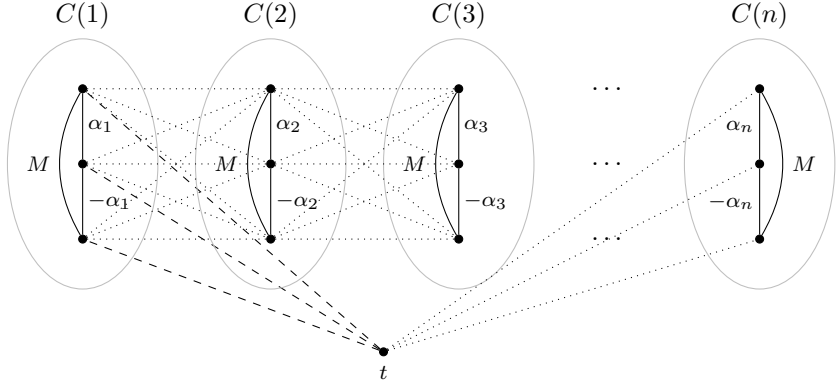

We now show that QTSP(p,c)-DEE (and hence QTSP(p,H)-DEE) can be solved in pseudopolynomial time and the problems admit FPTAS when the edge weights are non-negative. Our proof technique is to reduce QTSP(p,H)-DEE to QSPP(p,H) on a specially-constructed acyclic multigraph which we now describe.

Given a $G^*$ graph construct the acyclic digraph $G^{\prime}$  as follows. Note that the vertex $V^k$ of cycle $C(k)$ in $G^*$ is represented by $V^k =\{v^k_1,v^k_2,\ldots ,v^k_{r_k}\}$. Also, the edge set of $C(k)$ is $E(k) = \{e^k_1,e^k_2,\ldots ,e^k_{r_k}\}$, where $e^k_i = (v^k_i,v^k_{i+1})$ and the indices are taken modulo $r_k$. For $k=1,2,\ldots ,m-1$, create $\hat{E}(k) =\{\hat{e}^k_1, \hat{e}^k_2,\ldots ,\hat{e}^k_{r_k}\}$.  $\hat{E}(k)$ can be viewed as a copy of $E(k)$. Construct a graph $G^{\prime}=(V^{\prime},E^{\prime})$ where $V^{\prime}= \{t_1,t_2\}\cup E(m) \cup \left\{\cup_{k=1}^{m-1} (E(k)\cup\hat{E}(k))\right \}$. For each $k=1,2,\ldots ,m-1$ and $i,j=1,2,\ldots ,r_k$, introduce a directed edge $e=(e^k_i,\hat{e}^k_j), i\neq j$ and set $2p$ weights $\alpha^h_{e}=C(a^h,k)-a^h_{e^k_i}-a^h_{e^k_j}$ and $\beta^h_e= C(b^h,k)-b^h_{e^k_i}-b^h_{e^k_j}$ for $h=1,2,\ldots ,p$, where $C(a^h,k)=\sum_{e\in C(k)}a^h_e$ and $C(b^h,k)=\sum_{e\in C(k)}b^h_e$. The edge $e=(e^k_i,\hat{e}^k_j)$, $i\neq j$, represents the events of ejecting edges $e^k_i$ and $e^k_j$ from cycle $C(k)$ where a Hamiltonian cycle ``enters" $C(k)$ through $e^k_i$ and ``leaves" $C(k)$ through $e^k_j$. For every $k=1,2,\ldots ,m-1$, $i=1,2,\ldots ,r_{k}$, and $j=1,2,\ldots ,r_{k+1}$, introduce two directed edges $e_1=(\hat{e}^k_i,e^{k+1}_j)$ and $e_2=(\hat{e}^k_i,e^{k+1}_j)$. Note that $e_1$ and $e_2$ are parallel edges in $G^{\prime}$ in the same direction.  Let $u_1$ and $u_2$ be the endpoints of $e^k_i$ in $G^{*}$ and $v_1,v_2$ be the endpoints of $e^{k+1}_j$ in $G^*$.   Now set the weights $\alpha^h_{e_1}=a^h_{u_1v_1}+a^h_{u_2v_2}$ and $\beta^h_e= b^h_{u_1v_2}+b^h_{u_2v_1}$ for $h=1,2,\ldots ,p$. The edge $e_1$ represents ejecting $e^k_i$ from $C(k)$ and $e^{k+1}_j$ from $C(k+1)$  in $G^*$ and patching cycles using ``non-cross edges".  \myred{The edge $e_2$ represents the same event but the patching is done using ``cross edges" instead.} The tip vertex $t_1$ is connected to $e^1_i$ for $i=1,2,\ldots ,r_1$, and set $2p$ weights for the edges $e_i=(t, e^k_i)$ in $G^{\prime}$ as $\alpha^h(e_i)= a^h_{(t,v_i)}+a^h_{(t,v_{i+1})}$, $\beta^h(e_i)= b^h_{(t,v_i)}+b^h_{(t,v_{i+1})}$, where $e_i=(v_i,v_{i+1})$ in $G^*$. Finally, connect all the nodes $e^m_i$ for $i=1,2,\ldots ,r_m$, to $t_2$ and all the $\alpha$ and $\beta$ weights of these edges are zero.  The graph $G'$ constructed from the $G^*$ in Figure~\ref{fig:G-Star} is shown in Figure~\ref{fig:G'}.

\begin{figure}[ht]
\centering
\begin{tikzpicture}[->,>=stealth, decoration={markings, mark=at position 0.5 with {\arrow{>}}}]
\foreach \x[count=\xi] in {0.5,1.5,...,4}{
    \node[vrtx]  (e\xi) at (0,\x) {};
    \node[vrtx]  (f\xi) at (2.25,\x) {};}
\node[fit=(e4) (e3) (e2) (e1),ellipse,draw=gray!50,minimum width=2cm, label=above:{$E(1)$}] {};
\node[fit=(f4) (f3) (f2) (f1),ellipse,draw=gray!50,minimum width=2cm, label=above:{$\hat{E}(1)$}] {};
\foreach \x[count=\xi] in {0.5,1.5,...,3}{
    \node[vrtx]  (g\xi) at (4.5,.5+\x) {};
    \node[vrtx]  (h\xi) at (6.75,.5+\x) {};}
\node[fit=(g3) (g2) (g1),ellipse,draw=gray!50,minimum width=2cm, label=above:{$E(2)$}] {};
\node[fit=(h3) (h2) (h1),ellipse,draw=gray!50,minimum width=2cm, label=above:{$\hat{E}(2)$}] {};
\foreach \x[count=\xi] in {0.5,1.5,...,5}{
    \node[vrtx]  (i\xi) at (9,-.25+\x) {};
    \node[vrtx]  (j\xi) at (11.25,-.25+\x) {};}
\node[fit=(i1) (i2) (i3) (i4) (i5),ellipse,draw=gray!50,minimum width=2cm, label=above:{$E(3)$}] {};
\node[fit=(j1) (j2) (j3) (j4) (j5),ellipse,draw=gray!50,minimum width=2cm, label=above:{$\hat{E}(3)$}] {};
\foreach \x in {1,2,3,4}{
    \foreach \y in {1,2,3,4}{
        \ifnum\x=\y [\else \draw[dashed] (e\x) -- (f\y);] \fi}}
\foreach \x in {1,2,3,4}{
    \foreach \y in {1,2,3}{
        \draw[bend right,dotted] (f\x) to (g\y);
        \draw[bend left,dotted] (f\x) to (g\y);}}
\foreach \x in {1,2,3}{
    \foreach \y in {1,2,3}{
        \ifnum\x=\y [\else \draw[dashed] (g\x) -- (h\y);] \fi}}
\foreach \x in {1,2,3}{
    \foreach \y in {1,2,3,4,5}{
        \draw[bend right,dotted] (h\x) to (i\y);
        \draw[bend left,dotted] (h\x) to (i\y);}}
\foreach \x in {1,2,3,4,5}{
    \foreach \y in {1,2,3,4,5}{
        \ifnum\x=\y [\else \draw[dashed] (i\x) -- (j\y);] \fi}}
\node[vrtx, label=left:{\footnotesize $t_1$}]  (t1) at (-1.5,2) {};
\node[vrtx, label=right:{\footnotesize $t_2$}]  (t2) at (12.75,2.25) {};
\foreach \y in {1,2,3,4}{
    \draw[dashed] (t1) -- (e\y);}
\foreach \y in {1,2,3,4,5}{
    \draw[dashed] (j\y) -- (t2);}
\end{tikzpicture}
\caption{$G'$ constructed from the graph $G^*$ given in Figure \ref{fig:G-Star}.}
\label{fig:G'}
\end{figure}

\begin{theorem}
	From an optimal solution or $\epsilon$-optimal solution of QSPP(p,H,G$'$), an optimal solution to QTSP(p,H)-DEE can be recovered in linear time.
\end{theorem}
\begin{proof}
	From the construction of $G'$, it can be verified that there is a one-to-one correspondence between SEE-tours in $G^*$ and $t_1-t_2$ paths in $G'$ that preserves the objective function values of the corresponding solutions of QTSP(p,H)-DEE and QSPP(p,H).  Note that $G^{\prime}$ is an acyclic multigraph with at most two multiples of each edge. It is possible to \myred{use} our algorithm for QSPP(p,H) on an acyclic digraph \myred{for} the multigraph case as well, and the result follows.
\end{proof}
Now, from the construction above, and the results from the previous section, we immediately have the following.
\begin{corollary}
QTSP(p,H)-DEE can be solved in $O(mn^{2p+1}U)$ time, where $U=\prod_{h=1}^p \max_{e\in E} |a^h_e| \max_e |b^h_e|$, for any fixed $p$.  Moreover, QTSP(p,H)-DEE admits an FPTAS when $\vect{a},\vect{b} \geq \vect{0}$.
\end{corollary}

The instance of QTSP(A) when the family of tours is restricted to $F(DEE)$ is denoted by QTSP(A)-DEE.  Our reduction of QTSP(p,H)-DEE to QSPP(p,H) discussed above cannot be applied directly to solve QTSP(A)-DEE.  As before, the reduction can be modified to take into consideration the cost arising from adjacent pairs of edges to get an instance of adjacent QSPP(p,H) on an acyclic graph, and hence QTSP(A)-DEE can be solved in polynomial time.  We present a simple $O(n^3)$ algorithm to solve QTSP(A)-DEE directly.

Every DEE-tour in $G^*$ is defined by the edges which are removed upon entering and exiting each cycle $C(i)$, the edge which is removed from $C(m)$, and the choice of matching between the endpoints of the edge removed when exiting $C(i)$ and the edge removed when entering $C(i+1)$ for $i=1,\ldots,m-1$.  When the edges which are removed from cycle $C(i)$ share an endpoint, the quadratic costs which are incurred depend on the edges which are removed from $C(i-1)$ and $C(i+1)$ (as in the tour in Figure \ref{fig:G-Star-DEE}).  This prevents the approach used by Glover and Punnen \cite{glover1997travelling} for the linear TSP from being extended to QTSP(p,H)-DEE.  This also complicates any dynamic programming approach which attempts to construct an optimal solution by considering one cycle in each iteration, however, we show that it still can be done by considering two consecutive cycles instead of one in a dynamic programming recursion.

Let $f^1(v^k_i, v^{k-1}_{j})$ and $f^2(v^k_i, v^{k-1}_{j})$ be the lengths of the smallest expanded cycle $D(k)$ in $G^*$ containing edges $(v^k_i,v^{k-1}_{j})$, $(v^k_{i+1},v^{k-1}_{j+1})$, and $(v^k_i,v^{k-1}_{j+1})$, $(v^k_{i+1},v^{k-1}_{j})$, respectively, and let
\[
g\left(v^k_i\right) = \sum_{s=1}^{r_k}q\left(v^k_s,v^k_{s+1},v^k_{s+2}\right)-q\left (v^k_i,v^k_{i+1},v^k_{i+2}\right )-q\left (v^k_{i-1},v^k_i,v^k_{i+1}\right ).
\]
In the above expression and that follows, we assume that the indices $r_k+1 \equiv 1, r_k+2 \equiv 2,$ and $ 0 \equiv r_k$. \myred{Let $h^1_1\left(v^k_i,v^{k-1}_{j}\right)$ represent the length of the smallest expanded cycle $D(k)$ in $G^*$ containing edges $(v^k_i,v^{k-1}_{j})$ and $(v^k_{i+1},v^{k-1}_{j+1})$, that was not constructed by selecting two adjacent edges to eject from $C(k-1)$.  That is,}
\begin{align*}
h^1_1\left(v^k_i,v^{k-1}_{j}\right) =& q\left(v^k_{i},v^{k-1}_{j},v^{k-1}_{j-1}\right) + q\left(v^k_{i+1},v^{k-1}_{j+1},v^{k-1}_{j+2}\right) - q\left(v^{k-1}_{j-1},v^{k-1}_{j},v^{k-1}_{j+1}\right) - q\left(v^{k-1}_{j},v^{k-1}_{j+1},v^{k-1}_{j+2}\right) \\& + \min_{\substack{1\leq s \leq r_{k-1},\\s\not\in \{ j-1,j,j+1\} \\ 1\leq t\leq r_{k-2}}}  \left\{ f^1(v_s^{k-1},v_t^{k-2}),f^2(v_s^{k-1},v_t^{k-2}) \right\}.
\end{align*}
\myred{Let $h^1_2\left(v^k_i,v^{k-1}_{j}\right)$ represent the length of the smallest expanded cycle $D(k)$ containing edges $(v^k_i,v^{k-1}_{j})$ and $(v^k_{i+1},v^{k-1}_{j+1})$, that was constructed by selecting edge $(v^{k-1}_{j+1},v^{k-1}_{j+2})$ when constructing $D(k-1)$ and $(v^{k-1}_j,v^{k-1}_{j+1})$ when constructing $D(k)$.  That is,}
\begin{align*}
h^1_2\left(v^k_i,v^{k-1}_{j}\right) =& \myred{q\left(v^k_{i},v^{k-1}_{j},v^{k-1}_{j-1}\right) - q\left(v^{k-1}_{j-1},v^{k-1}_{j},v^{k-1}_{j+1}\right)} \\& + \min_{1\leq t\leq r_{k-2}} \{ f^1(v^{k-1}_{j+1},v^{k-2}_{t}) + q(v^{k}_{i+1},v^{k-1}_{j+1},v^{k-2}_{t})  - q(v^{k-1}_{j},v^{k-1}_{j+1}, v^{k-2}_{t}), \\& f^2(v^{k-1}_{j+1},v^{k-2}_{t}) + q(v^{k}_{i+1},v^{k-1}_{j+1},v^{k-2}_{t+1}) - q(v^{k-1}_{j},v^{k-1}_{j+1},v^{k-2}_{t+1}) \}.
\end{align*}
\myred{Similarly, let $h^1_3\left(v^k_i,v^{k-1}_{j}\right)$ represent the length of the smallest expanded cycle $D(k)$ containing edges $(v^k_i,v^{k-1}_{j})$ and $(v^k_{i+1},v^{k-1}_{j+1})$, that was constructed by selecting edge $(v^{k-1}_{j-1},v^{k-1}_{j})$ when constructing $D(k-1)$ and $(v^{k-1}_j,v^{k-1}_{j+1})$ when constructing $D(k)$.  That is,}
\begin{align*} h^1_3\left(v^k_i,v^{k-1}_{j}\right) =& \myred{q\left(v^k_{i+1},v^{k-1}_{j+1},v^{k-1}_{j+2}\right) - q\left(v^{k-1}_{j},v^{k-1}_{j+1},v^{k-1}_{j+2}\right)}
\\& + \min_{1\leq t\leq r_{k-2}} \{ f^1(v^{k-1}_{j-1},v^{k-2}_{t}) + q(v^{k}_{i},v^{k-1}_{j},v^{k-2}_{t+1}) - q(v^{k-1}_{j+1},v^{k-1}_{j}, v^{k-2}_{t+1}), \\& f^2(v^{k-1}_{j-1},v^{k-2}_{t}) + q(v^{k}_{i},v^{k-1}_{j},v^{k-2}_{t}) - q(v^{k-1}_{j+1},v^{k-1}_{j},v^{k-2}_{t}) \} .
\end{align*}
Then for $k=3,4,\ldots,m$,
\begin{align*}
f^1\left(v^k_i,v^{k-1}_{j}\right) =& g(v^k_i) + q(v^{k}_{i-1},v^{k}_{i},v^{k-1}_{j}) + q(v^{k}_{i+2},v^{k}_{i+1},v^{k-1}_{j+1}) \\& + \min\{ h^1_1\left(v^k_i,v^{k-1}_{j}\right), h^1_2\left(v^k_i,v^{k-1}_{j}\right), h^1_3\left(v^k_i,v^{k-1}_{j}\right) \}.
\end{align*}
Similar expressions follow for $h^2_1\left(v^k_i,v^{k-1}_{j}\right)$, $h^2_2\left(v^k_i,v^{k-1}_{j}\right)$, $h^2_3\left(v^k_i,v^{k-1}_{j}\right)$ and $f^2\left(v^k_i,v^{k-1}_{j}\right)$. The values of $f^1\left(v^2_i,v^1_{j}\right)$ and $f^2\left(v^2_i,v^1_{j}\right)$ for $1\leq i\leq r_2$ and $1\leq j\leq r_1$ can be calculated directly to initiate the above recursion. Thus, we can compute the value of the smallest expanded cycle $D(i)$ for $i=1,2,\ldots,r_m$. The foregoing discussions can be summarized in the theorem below.

\begin{theorem}
	QTSP(A)-DEE can be solved in $O(n^3)$ time.
\end{theorem}

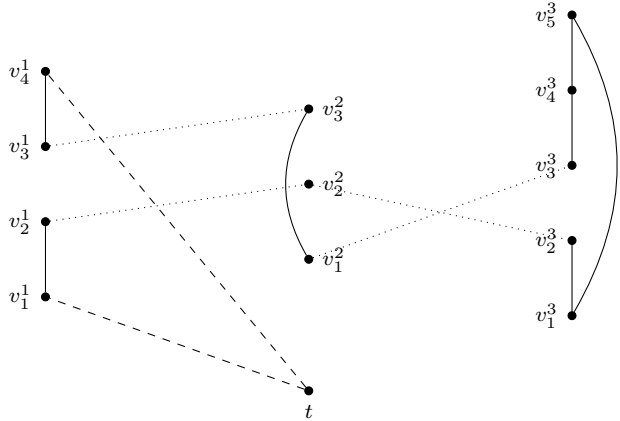
\begin{figure}[H]
    \centering
\begin{tikzpicture}
\foreach \x[count=\xi] in {0.5,1.5,...,4}{
\node[vrtx, label=left:{\footnotesize $v^1_\xi$}]  (v\xi) at (2,\x) {};
}
\foreach \x[count=\xi] in {0.5,1.5,...,3}{
\node[vrtx, label=right:{\footnotesize $v^2_\xi$}]  (w\xi) at (5.5,.5+\x) {};
}
\foreach \x[count=\xi] in {0.5,1.5,...,5}{
\node[vrtx, label=left:{\footnotesize $v^3_\xi$}]  (u\xi) at (9,-.25+\x) {};
}
\draw (v1) -- (v2)  (v3) -- (v4);
\draw[bend right] (w3) to (w1);
\draw (u1) -- (u2)  (u3) -- (u4) -- (u5);
\draw[bend left] (u5) to (u1);
\node[vrtx, label=below:{\footnotesize $t$}]  (t) at (5.5,-0.75) {};
\draw[dashed] (t) -- (v4) (t) -- (v1);
\draw[dotted] (v2) -- (w2) (v3) -- (w3);
\draw[dotted] (w2) -- (u2);
\draw[dotted] (w1) -- (u3);
\end{tikzpicture}
	\caption{An example of an optimal expanded cycle $D(3)$.  The cost can be computed as $f^2(v^3_2,v^2_1)=f^1(v^2_2,v^1_2) + g(v^3_2) + q(v^3_2,v^2_2,v^1_2) + q(v^3_1,v^3_2,v^2_2) + q(v^3_4,v^3_3,v^2_1) + q(v^3_3,v^2_1,v^2_3) - q(v^2_1,v^2_2,v^2_3) - q(v^2_2,v^2_1,v^2_3)$.  Note that some of the quadratic costs may contain vertices in 3 consecutive partitions of $G^*$, such as $q(v_2^1,v^2_2,v^3_2)$.}
	   \label{fig:QTSP(A)-DEE}
\end{figure}

\section{Paired vertex graphs} \label{sec:PV}

We now consider a class of undirected graphs which contains an exponential number of tours but on which the linear TSP is solvable in $O(n)$ time.  Let $G^{p}=(V,E)$ be constructed as follows.  Consider the sets $V^1,V^2,\ldots, V^\frac{n}{2}$ of pairs of vertices.  For each vertex in $V^k$, add an edge connecting it to every vertex in $V^{k+1}$, for all $k=1,2,\ldots,\frac{n}{2}-1$.  Add an edge connecting the two vertices in $V^1$ to each other, and the two vertices in $V^\frac{n}{2}$ to each other.  For $G^p$ with an odd number of vertices, a vertex can be added on the edge contained in $V^1$ and all following results hold.  We note that although this graph class is similar to the graph $G^*$, it is not a special case of $G^*$ and, to the best of our knowledge, has not been previously studied in connection with the linear TSP.

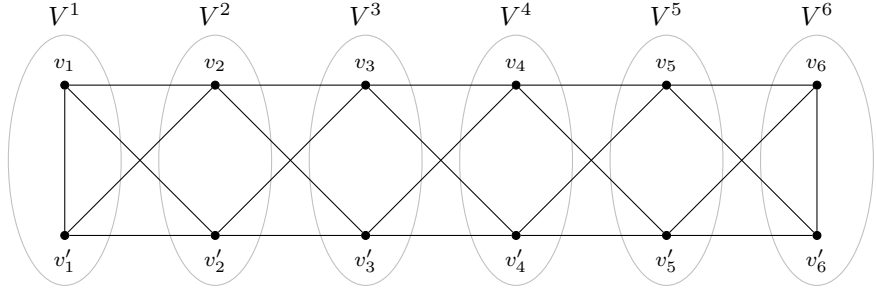
\begin{figure}[ht]
\centering
\begin{tikzpicture}
\foreach \x in {1,2,...,6}{
    \node [vrtx, label=above:{\footnotesize $v_\x$}] (va\x) at (2*\x, 2) {};]
    \node [vrtx, label=below:{\footnotesize $v'_\x$}] (vb\x) at (2*\x, 0) {};]
    \node[fit=(va\x) (vb\x),ellipse,draw=gray!50,minimum width=1.5cm, label=above:{$V^\x$}] {};
}
\draw (va1) to node {} (vb1);
\draw (va6) to node {} (vb6);
\draw (va1) -- (va2) -- (va3) -- (va4) -- (va5) -- (va6);
\draw (vb1) -- (vb2) -- (vb3) -- (vb4) -- (vb5) -- (vb6);
\draw (va1) -- (vb2) -- (va3) -- (vb4) -- (va5) -- (vb6);
\draw (vb1) -- (va2) -- (vb3) -- (va4) -- (vb5) -- (va6);
\end{tikzpicture}
\caption{Example of $G^p$ on 12 vertices.}
\label{fig:Gp}
\end{figure}

Let $F(PV)$ be the family of all tours which belong to $G^{p}$.  It can be verified that $|F(PV)|=2^{n/2-1}$.

\begin{theorem}
	The linear TSP on $G^p$ can be solved in $O(n)$ time.
\end{theorem}
\begin{proof}
	Let $G^p$ be a paired vertex graph on \myred{$2r=n$} vertices.  Every tour $\tau$ in $G^p$ contains the edges $(v_1,v'_1)$ and $(v_r,v'_r)$.  To connect $V^k$ to $V^{k+1}$, $\tau$ must either contain pairs of edges $(v_k,v_{k+1})$ and $(v'_k,v'_{k+1})$, or $(v_k,v'_{k+1})$ and $(v'_k,v_{k+1})$.  It is now clear that $\tau^*$ can be constructed greedily by adding pairs of edges joining vertices in $V^k$ to vertices in $V^{k+1}$ which minimize $\{c(v_k,v_{k+1}) + c(v'_k,v'_{k+1}), c(v_k,v'_{k+1}) + c(v'_k,v_{k+1})\}$ for each $k=1,\ldots,r-1$.
\end{proof}

Interestingly, the convex hull of the incidence vectors of tours in $F(PV)$ has a compact representation.
We give a linear description, $P(G^p)$ of the polytope of $G^p$.
\begin{theorem} \label{thm:Pgp}
	\begin{eqnarray}
		P(G^p) &=& \{ \vect{x}\in \mathbb{R^E}: 0\leq x_e\leq 1 \text{ for all } e\in E, \\
			&& x_{u_1,v} + x_{u_2,v}=1:u_1,u_2\in V^{k-1},v\in V^k, \text{ for all } k=2,\ldots,\frac{n}{2}, \label{eqn:gp1}\\
			&& x_{u_1,v} + x_{u_2,v}=1:v\in V^k,u_1,u_2\in V^{k+1}, \text{ for all } k=1,\ldots,\frac{n}{2}-1,\label{eqn:gp2}\\
			&& x_{u,v}=1: u,v\in V^1,\label{eqn:gp3}\\
			&& x_{u,v}=1: u,v\in V^\frac{n}{2} \label{eqn:gp4}\}.
	\end{eqnarray}
\end{theorem}
\begin{proof}
Let $A$ be the coefficient matrix for $P(G^p)$ and $\tau$ be the tour with characteristic vector $\vect{x}$.  Adding (\ref{eqn:gp1}) and (\ref{eqn:gp2}) implies that every vertex in $V^2,V^3,\ldots,V^{n/2-1}$ has degree 2 in $\tau$.  Since every edge in $G^p$ other than the edges contained in $V^1$ and $V^{n/2}$ connects vertices in successive partitions, a solution that contains a subtour must also include both edges incident with $v\in V^k$ and the vertices in $V^{k+1}$ (or $V^{k-1}$).  This contradicts either (\ref{eqn:gp1}) or (\ref{eqn:gp2}), and thus, $\tau$ is a tour in $G^p$.

$A$ is a binary matrix with exactly two ones in each row.  Moreover, since the variable for every edge is in exactly two constraints, there are exactly two 1's in each column.  It follows that the coefficient matrix \myred{is} totally unimodular, and hence $P(G^p)$ is a linear description of the polytope.
\end{proof}
The variant of QTSP when the tours are restricted to PV-tours is denoted QTSP-PV.

\begin{theorem}\label{thm:QTSP-PV}
	QTSP-PV is strongly NP-hard.
\end{theorem}
\begin{proof}
We reduce UBQP to QTSP-PV.  From an instance of UBQP on $n$ variables, we construct an instance of DQTSP-PV as follows.  Let $G^p$ be a graph on $n+1$ pairs of vertices, $V(G^p)=\cup_{k=1}^{n+1} V^k$, where $V^k=\{v_k,v'_k\}$.  $G^p$ contains edges connecting each vertex in $V^k$ to each vertex in $V^{k+1}$ for each $k=1,2,\ldots,n$, an edge connecting the vertices in $V^1$, as well as an edge connecting the vertices of $V^{n+1}$.  Assign costs $q((v_i,v_{i+1}),(v_j,v_{j+1}))=Q_{ij}$ for all $i,j$.  All other pairs of edges are assigned $q(e,f)=0$.

Given any solution $\vect{x}= (x_1,x_2,\ldots ,x_n)$ of UBQP, we can construct a tour $\tau$ in $G^p$ containing the edges $(v_i,v_{i+1})$ and $(v'_i,v'_{i+1})$ if $x_i=1$ and the edges $(v_i,v'_{i+1})$ and $(v'_{i},v_{i+1})$ if $x_i=0$, for $1\leq i\leq n$, as well as the edges contained in $V^1$ and $V^{n+1}$.  It can be verified that the cost of $\tau$ is precisely $\vect{x}^TQ\vect{x}$.

Conversely, given any tour $\tau$ in the $G^p$ obtained above, construct a vector $\vect{x}$ as $x_i=1$ if and only if edge $(v_i,v_{i+1})$ belongs to $\tau$.  The cost of the tour $\tau$ is precisely $\vect{x}^TQ\vect{x}$.  Since UBQP is strongly NP-hard, the proof follows.
\end{proof}

\begin{figure}[H]
\centering
\begin{tikzpicture}
\foreach \x in {1,2,...,6}{
    \node [vrtx, label=above:{\footnotesize $v_\x$}] (va\x) at (2*\x, 2) {};]
    \node [vrtx, label=below:{\footnotesize $v'_\x$}] (vb\x) at (2*\x, 0) {};]
    \node[fit=(va\x) (vb\x),ellipse,draw=gray!50,minimum width=1.5cm, label=above:{$V^\x$}] {};}
\draw (va1) to node {} (vb1);
\draw (va6) to node {} (vb6);
\draw (va1) -- (va2);
\draw (va4) -- (va5) -- (va6);
\draw (vb1) -- (vb2);
\draw (vb4) -- (vb5) -- (vb6);
\draw (vb2) -- (va3) -- (vb4);
\draw (va2) -- (vb3) -- (va4);
\end{tikzpicture}
\caption{Example of a tour $\tau$ in $G^p$ which corresponds to the solution $\vect{x}=(1,0,0,1,1)$ in the proof of Theorem~\ref{thm:QTSP-PV}.}
\label{fig:Gp}
\end{figure}

The problem QTSP(p,H) where the family of feasible solutions is restricted to PV-tours is called the paired vertex QTSP with rank $p$ and is denoted by QTSP(p,H)-PV.  We have the analogous definition for QTSP(p,c).

\begin{theorem}\label{QTSP(p,H)-PV}
	QTSP(p,c)-PV is NP-hard even when $p=1$ and $c(e)=0$ for all $e\in E$.
\end{theorem}
\begin{proof}
We reduce the PARTITION problem to QTSP(1,H)-PV.  From an instance of PARTITION \myred{with the given data $\alpha_1,\ldots,\alpha_n$}, we construct an instance of QTSP(1,H)-PV as follows.  Let $G^p$ be a graph on $n+1$ pairs of vertices, $V(G^p)=\cup_{k=1}^{n+1} V^k$, where $V^k=\{v_k,v'_k\}$.  $G^p$ contains edges connecting each vertex in $V^k$ to each vertex in $V^{k+1}$ for each $k=1,2,\ldots,n$, an edge connecting the vertices in $V^1$, as well as an edge connecting the vertices of $V^{n+1}$.  Assign costs $a(v_i,v_{i+1})=b(v_i,v_{i+1})=\myred{\alpha_i}$ and $a(v_i,v'_{i+1})=b(v_i,v'_{i+1})=-\myred{\alpha_i}$ for each $i=1,2,\ldots,n$.  The objective function of QTSP(1,H)-PV on the $G^p$ constructed above becomes $(\sum_{e\in \tau}a_e)^2 \geq 0$, where $\tau$ is a tour in this $G^p$.  It can be verified that the optimal objective function value of this QTSP(1,H)-PV is zero precisely when the required PARTITION exists.  The \myred{result} follows from the NP-completeness of PARTITION~\cite{Karp1972}.
\end{proof}

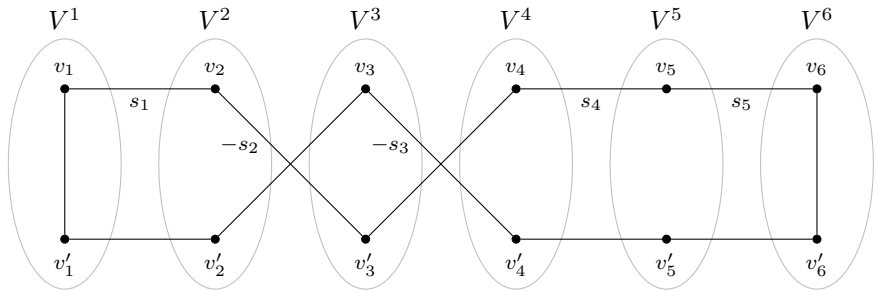
\begin{figure}[ht]
\centering
\begin{tikzpicture}
\foreach \x in {1,2,...,6}{
    \node [vrtx, label=above:{\footnotesize $v_\x$}] (va\x) at (2*\x, 2) {};]
    \node [vrtx, label=below:{\footnotesize $v'_\x$}] (vb\x) at (2*\x, 0) {};]
    \node[fit=(va\x) (vb\x),ellipse,draw=gray!50,minimum width=1.5cm, label=above:{$V^\x$}] {};}
\draw (va1) to node {} (vb1);
\draw (va6) to node {} (vb6);
\draw (va1) -- node[below] {\footnotesize $s_1$} ++ (va2);
\draw (va4) -- node[below] {\footnotesize $s_4$} (va5) -- node[below] {\footnotesize $s_5$} ++ (va6);
\draw (vb1) -- (vb2);
\draw (vb4) -- (vb5) -- (vb6);
\draw (vb2) -- (va3) -- node[below,xshift=-0.67cm,yshift=0.48cm]{\footnotesize$-s_3$} (vb4);
\draw (va2) -- node[below,xshift=-0.67cm,yshift=0.48cm]{\footnotesize$-s_2$} (vb3) -- (va4);
\end{tikzpicture}
\caption{Example of a tour $\tau$ in $G^p$ which corresponds to the solution $S=\{1,4,5\}$ in the proof of Theorem~\ref{QTSP(p,H)-PV}.}
\label{fig:Gp}
\end{figure}

Despite this negative result, we now show that when $p$ is fixed, QTSP(p,H)-PV can be solved in pseudopolynomial time and in this case it also admits an FPTAS when the edge weights are nonnegative.  Recall that an instance of QTSP(p,H)-PV is given by $p$ pairs of costs $a^r_{ij},b^r_{ij}$ for \myred{$r=1,2,\ldots,p$}, for each edge $(i,j)\in G^p$.  We formulate QTSP(p,H) as a rank $p$ quadratic shortest path problem in an acyclic directed graph.

Given a graph $G^p$, construct the acyclic digraph $G'$ as follows.  Note that the vertex set $V^k=\{v_k,v'_k\}$ for $k=1,2,\ldots, \frac{n}{2}$.  Construct graph $G'=(V',E')$ where $V'=\{\hat{v}_1,\hat{v}_2\ldots,\hat{v}_{n/2}\}$.  For each pair of edges $(v_k,v_{k+1})$ and $(v'_k,v'_{k+1})$ in $G^p$, introduce a directed edge $e_k=(\hat{v}_k,\hat{v}_{k+1})$ which represents the edges $(v_k,v_{k+1})$ and $(v'_k,v'_{k+1})$ being included in a Hamiltonian cycle, and similarly, for each pair of edges $(v_k,v'_{k+1})$ and $(v'_k,v_{k+1})$ in $G^p$, introduce a directed edge $\bar{e}_k=(\hat{v}_k,\hat{v}_{k+1})$.  For $h=1,2,\ldots,p$, set $\alpha^h_{e_1}=a^h_{v_1,v'_1} + a^h_{v_1,v_2} + a^h_{v'_1,v'_2} + a^h_{v_k,v'_k}$, $\alpha^h_{\bar{e}_1}=a^h_{v_1,v'_1} + a^h_{v_1,v'_2} + a^h_{v'_1,v_2} + a^h_{v_k,v'_k}$, $\beta^h_{e_1}=b^h_{v_1,v'_1} + b^h_{v_1,v_2} + b^h_{v'_1,v'_2} + b^h_{v_k,v'_k}$, and $\beta^h_{\bar{e}_1}=b^h_{v_1,v'_1} + b^h_{v_1,v'_2} + b^h_{v'_1,v_2} + b^h_{v_k,v'_k}$.  For $k=2,3\ldots,n/2-1$, and $h=1,2,\ldots,p$ we set $\alpha^h_{e_k}=a^h_{v_k,v_{k+1}} + a^h_{v'_k,v'_{k+1}}$, $\alpha^h_{\bar{e}_k}=a^h_{v_k,v'_{k+1}} + a^h_{v'_k,v_{k+1}}$, $\beta^h_{e_k}=b^h_{v_k,v_{k+1}} + b^h_{v'_k,v'_{k+1}}$ and $\beta^h_{\bar{e}_k}=b^h_{v_k,v'_{k+1}} + b^h_{v'_k,v_{k+1}}$.  The graph $G'$ constructed from the $G^p$ in Figure~\ref{fig:Gp} is shown in Figure~\ref{fig:Gp-sp}.

\begin{figure}[ht]
\centering
\begin{tikzpicture}[-,>=stealth, decoration={markings, mark=at position 0.5 with {\arrow{>}}}]
\foreach \x in {1,2,...,6}{
    \node [vrtx, label=above:{\footnotesize $\hat{v}_\x$}] (v\x) at (2*\x, 2) {};]
}
\foreach \x[count=\xi] in {2,3,...,6}{
	\draw[bend right, dotted, postaction={decorate}] (v\xi) to (v\x);
	\draw[bend left, dotted,postaction={decorate}] (v\xi) to (v\x);}
\end{tikzpicture}
\caption{$G'$ constructed from the graph $G^p$ given in Figure~\ref{fig:Gp}.}
\label{fig:Gp-sp}
\end{figure}
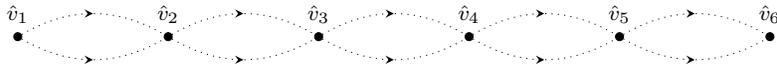

\myred{Let $\Pi$ be the collection of a paths from $\hat{v}_1$ to $\hat{v}_{n/2}$ in $G'$.} From the construction given above, it can be verified that there is a one-to-one correspondence between elements of $\Pi$ and $F(PV)$ where the corresponding elements have the same weight.  Further, given an element of $\Pi$, we can construct a corresponding element in $F(PV)$ in polynomial time.  Note that the graph $G^{\prime}$ is an acyclic multigraph with exactly two multiples of each edge.  Thus, QTSP(p,H)-PV can be solved in pseudopolynomial time, and a minor modification of the analysis in the proof of Theorem~\ref{thm:QSPP(p,H)} yields the  following theorem.  Although the number of edges doubles, this does not change the worst-case complexity.
\begin{corollary}
    QTSP(p,H)-PV can be solved in $O(n^{2p+1}U)$ time, where $U=\prod_{h=1}^p (\max_{e\in E} |a^h_e| \max_e |b^h_e|)$, for any fixed $p$.  Moreover, QTSP(p,H)-PV admits an FPTAS when $\vect{a},\vect{b} \geq \vect{0}$.
\end{corollary}

We now show that the adjacent quadratic TSP restricted to the set of paired vertex tours, denoted QTSP(A)-PV, can be solved in polynomial time using dynamic programming.  The input for QTSP(A)-PV is given as the costs of paths of length two in $G^p$.  \myred{That is,} for any $2$-path $u-v-w$ with $v$ as the middle vertex, a cost $q(u,v,w)$ is given.  Note that $q(u,v,w)=q(w,v,u)$.  Let $f(k)$ be the length of the smallest PV-Hamiltonian path in $G^p$ from $v_k$ to $v'_k$ containing the edges \myred{$(v_{k-1},v_k)$} and \myred{$(v'_{k-1},v'_k)$}.  Similarly, let $g(k)$ be the length of the smallest PV-Hamiltonian path containing \myred{$(v_{k-1},v'_{k})$} and \myred{$(v'_{k-1},v_k)$}.  Then for $k=2,3,\ldots,\frac{n}{2}-1$,
\begin{align*}
    f(k+1) = \min\{ f(k) + q(v_{k-1},v_k,v_{k+1}) + q(v'_{k-1},v'_k,v'_{k+1}), g(k) + q(v_{k-1},v'_k,v'_{k+1}) + q(v'_{k-1},v_k,v_{k+1})\} \},
\end{align*}
\myred{and}
\myred{
\begin{align*}
    g(k+1) = \min\{ f(k) + q(v_{k-1},v_k,v'_{k+1}) + q(v'_{k-1},v'_k,v_{k+1}), g(k) + q(v_{k-1},v'_k,v_{k+1}) + q(v'_{k-1},v_k,v'_{k+1})\} \}.
\end{align*}}
The values of $f(2)$ and $g(2)$ can be calculated directly to initiate the recursion.  Adding the edge $(v_{\frac{n}{2}},v'_{\frac{n}{2}})$ completes the tour and the better of the two tours gives an optimal solution to QTSP(A) on $G^p$.  The foregoing discussion can be summarized in the theorem below.
\begin{theorem}
	QTSP(A)-PV can be solved in $O(n)$ time.
\end{theorem}
The results discussed in this section can easily be modified to obtain corresponding results when $n$ is odd by adding a single vertex $v_1''$ along the edge $(v_1,v_1')$.  Then, every tour which contains $(v_1,v_1')$ will contain edges $(v_1,v_1'')$ and $(v_1',v_1'')$ in the modified graph.  The results for the case when $n$ is odd follow.

\section{Matching Edge Ejection Tours} \label{sec:NI}

In this section we consider a special class of tours considered by Punnen~\cite{Punnen2001}.  Consider a special spanning subgraph $G^M$ of the complete graph $K_n$ obtained as follows.  Partition the vertices of $K_n$ into two sets, $U=\{u_1,u_2,\ldots,u_r\}$ and $V=\{v_1,v_2,\ldots,v_s\}$.  Let $E_u=\{(u_i,u_{i+1}):1\leq i\leq r\}$, where $r+1\equiv 1$ and $E_{uv}=\{(u_i,v_j):1\leq i\leq r, 1\leq j\leq s\}$.  Hereafter, we assume that $s\leq r$.  The edge set of $G^M$ is defined as $E(G^M)=E_u\cup E_{uv}$. The resulting graph is denoted by $G^M=(V^M,E^M)$, (See Figure \ref{fig:G-M} for an example of a $G^M$ graph), where $V^M=U\cup V$.

\begin{figure}[H]
\centering
\begin{tikzpicture}[scale=0.6]
\foreach \x[count=\xi] in {0.5,1.3,...,5.2}{
    \node[vrtx]  (u\xi) at (0,\x) {};}
\node[fit=(u6) (u1),ellipse,draw=gray!50,minimum width=2cm, label=above:{$U$}] {};
\foreach \x[count=\xi] in {0.5,1.3,...,3.6}{
    \node[vrtx]  (v\xi) at (4,0.8+\x) {};}
\node[fit=(v1) (v4),ellipse,draw=gray!50,minimum width=2cm, label=above:{$V$}] {};
\draw (u1) -- (u2) -- (u3) -- (u4) -- (u5) -- (u6);
\draw[bend right] (u6) to (u1);
\foreach \x in {1,2,3,4,5,6}{
    \foreach \y in {1,2,3,4}{
        \draw[dotted] (u\x) -- (v\y);}}
\end{tikzpicture}
	\caption{A graph $G^M$ with $r=6$ and $s=4$.}
	   \label{fig:G-M}
\end{figure}
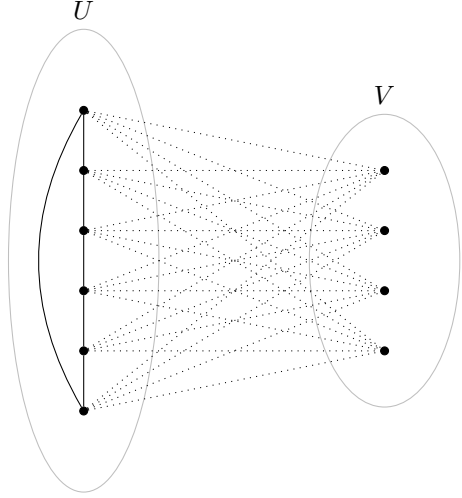

Since the TSP is NP-hard on a complete bipartite graph and $G^M$ has a spanning subgraph which is a complete bipartite graph, the TSP is NP-hard on $G^M$ as well.  Let us now consider a family of tours in $G^M$, called \emph{matching edge ejection tours} (MEE-tours) which consists of all tours in $G^M$ that can be obtained by the following process.
\begin{enumerate}
    \item Eject $s$ edges $e_{\pi(1)},e_{\pi(2)},\ldots,e_{\pi(s)}$ from the cycle $E_u\equiv (u_1,u_2,\ldots,u_r,u_1)$, and let $E_u(s)=E_u - \{e_{\pi(1)},e_{\pi(2)},\ldots,e_{\pi(s)}\}$ be the edge set of the resulting subgraph.
    \item Insert the vertices $v_i\in V$ into $E_u(s)$ by connecting it by edges to the endpoints of $e_{\pi(i)}$, $1\leq i\leq s$ to construct a tour in $G^M$. (See Figure \ref{fig:MEE-tour} for an MEE-tour in the $G^M$ graph of Figure \ref{fig:G-M}).
\end{enumerate}

\begin{figure}[H]
\centering
\begin{tikzpicture}[scale=0.6]
\foreach \x[count=\xi] in {0.5,1.4,...,5.4}{
    \node[vrtx]  (u\xi) at (0,\x) {};}
\node[fit=(u6) (u1),ellipse,draw=gray!50,minimum width=2cm, label=above:{$U$}] {};
\foreach \x[count=\xi] in {0.5,1.4,...,4}{
    \node[vrtx]  (v\xi) at (4,0.9+\x) {};}
\node[fit=(v1) (v4),ellipse,draw=gray!50,minimum width=2cm, label=above:{$V$}] {};
\draw (u3) -- (u4);
\draw (u5) -- (u6);
\draw[dotted] (u2) -- (v1) -- (u3);
\draw[dotted] (u1) -- (v2) -- (u6);
\draw[dotted] (u1) -- (v3) -- (u2);
\draw[dotted] (u4) -- (v4) -- (u5);
\end{tikzpicture}
	\caption{An MEE-tour in the graph $G^M$ given in Figure \ref{fig:G-M}.}
	   \label{fig:MEE-tour}
\end{figure}
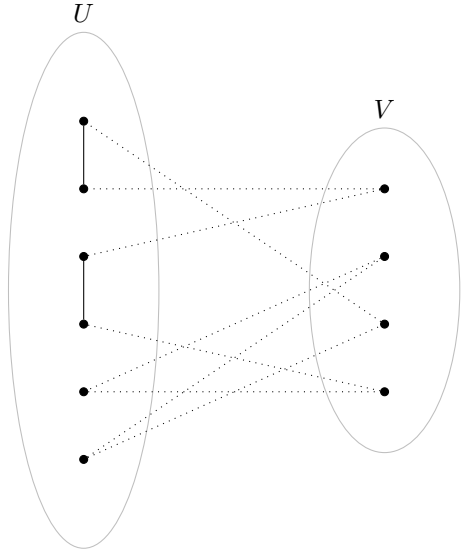


Let $F(MEE)$ be the collection of all MEE-tours in $G^M$. It has been shown in~\cite{Punnen2001} that $|F(MEE)|=\frac{r!}{(r-s)!}$.  If $n$ is even and $r=n/2$\myred{,} then $|F(MEE)|=(\frac{n}{2})!$\myred{,} and if $n$ is odd and $r=(n+1)/2$\myred{,} then $|F(MEE)|=(\frac{n+1}{2})!$.  In fact, $|F(MEE)|$ could be even larger than $(\frac{n}{2})!$ for an appropriate choice of $r$ and $s$.  Gutin and Yeo~\cite{Gutin1999} showed that $|F(MEE)|$ could be as large as $(\frac{n}{2} + p_0)!/(2p_0)!$ where $p_0=\sqrt{\frac{1}{8}(n+\frac{9}{8})} + \frac{3}{8}$.  Finding the best QTSP tour in $F(MEE)$ is a nontrivial task.  Interestingly, TSP restricted to MEE-tours can be solved in $O(n^3)$ time by formulating it as a minimum weight perfect matching problem on an associated bipartite graph~\cite{Punnen2001}.

The quadratic travelling salesman problem where the family of feasible solutions is restricted to MEE-tours is denoted by QTSP-MEE.  Note that by using the rank decomposition of the matrix $Q$ (which has rank $p$), QTSP-MEE can be stated as
\begin{align*}
\text{Minimize     }  q(\tau) =& \sum_{h=1}^p\left(\sum_{e\in \tau} a_e^h\right)\left(\sum_{e\in \tau}b_e^h\right) \\
\text{Subject to } & \tau \in F(SEE).
\end{align*}
It may be noted that in the above representation, $p$ could be $O(n^2)$.

QTSP-MEE is strongly NP-hard.  This follows directly from a stronger result that we later prove.

The quadratic assignment problem on the complete bipartite graph $G'=(U,V,E)$ which, by using the rank decomposition, can be stated as
\begin{eqnarray*}
    QAP(G'): &\textrm{Minimize } & q(P) = \sum_{h=1}^p \left(\sum_{e\in P}\alpha^h_e\right) \left( \sum_{e\in P}\beta^h_e\right) \\
    & \textrm{Subject to } & P \in \mathcal{P},
\end{eqnarray*}
where $\mathcal{P}$ is the set of all perfect matchings in $G'$, and $|U|=|V|=n$.  When $n$ is odd, we denote QAP(G$'$) as odd-QAP.  It is easy to see that odd-QAP is strongly NP-hard.

By extending the formulation~\cite{Punnen2001}, we can also formulate QTSP-MEE as a QAP(G$'$).

Given a graph $G^M$, construct the complete bipartite graph $G'$ as follows.  Note that the vertex set of $G^M$ is represented by $V^M=U\cup V$ where $U=\{u_1,u_2,\ldots,u_r\}$ and $V=\{v_1,v_2,\ldots,v_s\}$.  Also, the edge set is $E(G^M)=E_u\cup E_{uv}$ where $E_u=\{e_i=(u_i,u_{i+1}):1\leq i\leq r\}$ and $E_{uv}=\{(u_i,v_j):1\leq i\leq r, 1\leq j\leq s\}$, $r+1\equiv 1$ and $s\leq r$.  Construct a complete bipartite graph $G'=(V',E')$ where $V'=\{E_u \cup (V \cup \{v_i:s<i\leq r\}) \}$ and $E'=\{(e_i,v_j):e_i\in E_u,v_j\in V \cup \{v_i:s<i\leq r\}\}$.  For $j\in V$ set weights $\alpha_{ij}^h=a_{u_i,v_j}^h + a_{v_j,u_{i+1}}^h - a_{u_i,u_{i+1}}^h$ and $\beta^h_{ij}=b_{u_i,v_j}^h + b_{v_j,u_{i+1}}^h - b_{u_i,u_{i+1}}^h$, and set weights $\alpha_{ij}^h=a_{u_i,u_{i+1}}^h$ and $\beta^h_{ij}=b_{u_i,u_{i+1}}^h$, otherwise, for all $h=1,2,\ldots,p$ and \myred{all $i$ such that $e_i\in E_u$}.  For $j\leq s$, the edge $e=(e_i,v_j)$ represents the events of ejecting edge $e_i$ from cycle $E_u$ and inserting $v_j$ by joining it to the endpoints of $e_i$, otherwise $e$ represents the event that no vertex is inserted along $e_i$.

The problem QTSP(p,c) restricted to the collection of tours in $F(MEE)$ is denoted QTSP(p,c)-MEE.  We have the analogous definition for the homogenous case, denoted QTSP(p,H)-MEE.

\begin{corollary}\label{thm:QTSP(p,H)-MEE}
	QTSP(p,c)-MEE is NP-hard even if $p=1$ and $c(e)=0$ for all $e\in E^M$.
\end{corollary}
The proof follows from the reduction above using the fact that rank 1 odd QAP is NP-hard~\cite{Punnen01}.

\begin{corollary}
    QTSP(1,H)-MEE admits an FPTAS when $\vect{a},\vect{b}\geq \vect{0}$.
\end{corollary}
The proof of this corollary follows from the reduction given above and applying the result of Goyal et al.~\cite{Goyal2011} on the resulting rank 1 odd QAP.

The adjacent quadratic TSP over the collection of MEE-tours is denoted QTSP(A)-MEE.  Note that QTSP(A)-SEE, QTSP(A)-DEE, and QTSP(A)-PV are solvable in polynomial time.  This simplicity however, does not extend to QTSP(A)-MEE.

\begin{theorem}\label{thm:QTSP(A)-MEE}
	QTSP(A)-MEE is strongly NP-hard.
\end{theorem}
\begin{proof}
We give a reduction from the linear TSP.  Given a graph $G$ on the vertices $1,2,\ldots,n$, \myred{and linear cost function $c$ defined on the edges of $G$,} the graph $G^M$ is constructed on the vertex set $V^M=U\cup V$, where $U=\{u_1,u_2,\ldots, u_n\}$ and $V=\{v_1,v_2,\ldots, v_n\}$.  Let $E_u$ be the cycle $(u_1,u_2,\ldots,u_n,u_1)$.  The indices are taken modulo $n$.  Assign quadratic cost on the pairs of adjacent edges $q((u_i,v_j),(v_j,u_k))=c(i,k)$ for $i,j,k=1,2,\ldots,n$, $i\neq k$.  All other costs are zero.
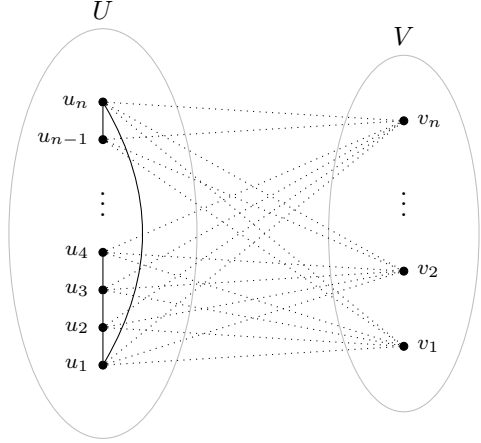
\begin{figure}[H]
\centering
\begin{tikzpicture}
\node[vrtx, label=left:{\footnotesize $u_1$}]  (u1) at (0,0.5) {};
\node[vrtx, label=left:{\footnotesize $u_2$}]  (u2) at (0,1) {};
\node[vrtx, label=left:{\footnotesize $u_3$}]  (u3) at (0,1.5) {};
\node[vrtx, label=left:{\footnotesize $u_4$}]  (u4) at (0,2) {};
\node[vrtx, label=left:{\footnotesize $u_{n-1}$}]  (unm1) at (0,3.5) {};
\node[vrtx, label=left:{\footnotesize $u_n$}]  (un) at (0,4) {};
\node[fit=(un) (u1),ellipse,draw=gray!50,minimum width=2.1cm, label=above:{$U$}] {};
\node[vrtx, label=right:{\footnotesize $v_1$}]  (v1) at (4,0.5) {};
\node[vrtx, label=right:{\footnotesize $v_2$}]  (v2) at (4,1) {};
\node[vrtx, label=right:{\footnotesize $v_3$}]  (v3) at (4,1.5) {};
\node[vrtx, label=right:{\footnotesize $v_4$}]  (v4) at (4,2) {};
\node[vrtx, label=right:{\footnotesize $v_{n-1}$}]  (vnm1) at (4,3.5) {};
\node[vrtx, label=right:{\footnotesize $v_n$}]  (vn) at (4,4) {};
\node[fit=(v1) (vn),ellipse,draw=gray!50,minimum width=2.1cm, label=above:{$V$}] {};
\draw (u1) -- (u2) -- (u3) -- (u4);
\draw (unm1) -- (un);
\draw[bend left] (un) to (u1);
\foreach \x in {1,2,3,4,nm1,n}{
	\foreach \y in {1,2,3,4,nm1,n}{
        		\draw[dotted] (u\x) -- (v\y);}}
\path (u4) -- node[auto=false]{\vdots} (unm1);
\path (v4) -- node[auto=false]{\vdots} (vnm1);
\end{tikzpicture}
	\caption{Construction used in the proof of Theorem~\ref{thm:QTSP(A)-MEE}.}
	   \label{fig:QTSP(A)-MEE-proof}
\end{figure}

It can be verified that every tour $\pi=(\pi(1),\pi(2),\ldots,\pi(n),\pi(1))\in G$ has the same cost as the tour $\pi'$ which results from inserting $\pi(i)$ into edge $(u_i,u_{i+1})$, for each $i=1,2,\ldots,n$, that is, \\$\pi'=(u_1,v_{\pi(1)},u_2,v_{\pi(2)},\ldots,u_n,v_{\pi(n)},u_1)\in G'$.  This establishes a 1-1 correspondence between tours in $G$ and MEE-tours in $G^M$, and the result follows.
\end{proof}

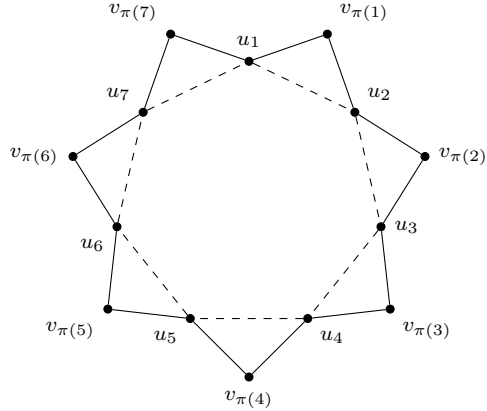
\begin{figure}[H]
	\centering
    \begin{tikzpicture}[scale=1.2]
        \node [vrtx, label=above:{\footnotesize $u_1$}] (v1) at (0, 1.5){};
	\node [vrtx, label=above right:{\footnotesize$u_2$}] (v2) at (1.1727, 0.9352) {};
        \node [vrtx, label=right:{\footnotesize$u_3$}] (v3) at (1.4624, -0.3338) {};
        \node [vrtx, label=below right:{\footnotesize$u_4$}] (v4) at (0.6508, -1.3515) {};
        \node [vrtx, label=below left:{\footnotesize$u_5$}] (v5) at (-0.6508, -1.3515) {};
        \node [vrtx, label=below left:{\footnotesize$u_6$}] (v6) at (-1.4624, -0.3338) {};
        \node [vrtx, label=above left:{\footnotesize$u_7$}] (v7) at (-1.1727, 0.9352) {};
        \node [vrtx, label=above right:{\footnotesize$v_{\pi(1)}$}] (sigma1) at (0.8678, 1.8019){};
        \node [vrtx, label=right:{\footnotesize$v_{\pi(2)}$}] (sigma2) at (1.9499, 0.4450){};
        \node [vrtx, label=below right:{\footnotesize$v_{\pi(3)}$}] (sigma3) at (1.5637, -1.2470){};
        \node [vrtx, label=below:{\footnotesize$v_{\pi(4)}$}] (sigma4) at (0, -2){};
        \node [vrtx, label=below left:{\footnotesize$v_{\pi(5)}$}] (sigma5) at (-1.5637, -1.2470){};
        \node [vrtx, label=left:{\footnotesize$v_{\pi(6)}$}] (sigma6) at (-1.9499, 0.4450){};
        \node [vrtx, label=above left:{\footnotesize$v_{\pi(7)}$}] (sigma7) at (-0.8678, 1.8019){};
	    \draw[dashed] (v1) -- (v2) -- (v3) -- (v4) -- (v5) -- (v6) -- (v7) -- (v1);
        \draw (v1) -- (sigma1) -- (v2) -- (sigma2) -- (v3) -- (sigma3) -- (v4) -- (sigma4)
        -- (v5) -- (sigma5) -- (v6) -- (sigma6) -- (v7) -- (sigma7) -- (v1);
 	\end{tikzpicture}
	\caption{An example of the construction used in the proof of Theorem~\ref{thm:QTSP(A)-MEE} is shown.  The solid lines indicate the tour $\pi'\in G'$ defined by $\pi\in G$.  The cycle \myred{$E_u=(u_1,u_2,\ldots,u_7,u_1)$} is shown with dashed lines.}
	   \label{fig:NI-A}
\end{figure}

\section{Conclusion}
We presented a systematic study of various complexity aspects of QTSP which generalizes the well-known travelling salesman problem.  We have shown that QTSP is NP-hard on several classes of exponential neighbourhoods for which the linear TSP is polynomially-solvable.  We introduce a restricted version of the QTSP objective, the fixed-rank QTSP, and examine the complexity of this problem on these classes of exponential neighbourhoods.  It is shown that QTSP(p,c)-SEE, QTSP(p,c)-DEE, and QTSP(p,c)-PV can be solved in pseudopolynomial time and they also admit FPTAS.  QTSP(p,c)-MEE with $p=1$ can be solved in pseudopolynomial time and admits an FPTAS.  For fixed $p>1$, the complexity status is open.  For the adjacent QTSP variation, i.e. QTSP(A)-SEE, QTSP(A)-DEE and QTSP(A)-PV, we present polynomial algorithms.  The problem QTSP(A)-MEE is shown to be NP-hard.  As a by-product, we obtain an FPTAS for the fixed-rank quadratic shortest path problem, and a pseudopolynomial algorithm when the problem is restricted to acyclic graphs.


\bibliographystyle{plain}

\end{document}